\documentclass[10pt,twoside,web]{ieeecolor}
\usepackage{generic}
\usepackage{bigints}
\usepackage{abraces}
\usepackage{bbold}
\usepackage[noadjust]{cite}

\usepackage{amsmath,amssymb,amsfonts}
\usepackage{algorithmic}
\usepackage{graphicx}
\usepackage{textcomp}
\usepackage{color}
\definecolor{darkgreen}{rgb}{0,0.4,0}
\usepackage{float}
\definecolor{darkgreen}{rgb}{0,0.6,0}
\usepackage{authblk}
\usepackage[colorlinks,linkcolor=blue,citecolor=darkgreen]{hyperref}
\usepackage[latin1]{inputenc}
\usepackage{amsmath}
\usepackage{graphicx}
\usepackage{amssymb}
\usepackage{verbatim}
\usepackage{epsfig}
\usepackage[labelfont=bf]{caption}
\usepackage{enumerate}
\usepackage{subfig}
\usepackage{epstopdf}
\usepackage{relsize,exscale}

\usepackage{booktabs} 
\newtheorem{theorem}{Theorem}
\newtheorem{definition}{Definition}
\newtheorem{assumption}{Assumption}
\newtheorem{lemma}{Lemma}
\newtheorem{remark}{Remark}

\newtheorem{proposition}{Proposition}

\def\begquo{\begin{quote}}
\def\endquo{\end{quote}}
\def\begequarr{\begin{eqnarray}}
\def\endequarr{\end{eqnarray}}
\def\begequarrs{\begin{eqnarray*}}
\def\endequarrs{\end{eqnarray*}}
\def\begarr{\begin{array}}
\def\endarr{\end{array}}
\def\begequ{\begin{equation}}
\def\endequ{\end{equation}}

\def\begdes{\begin{description}}
\def\enddes{\end{description}}
\def\begenu{\begin{enumerate}}
\def\begite{\begin{itemize}}
\def\endite{\end{itemize}}
\def\endenu{\end{enumerate}}

\def\lef[{\left[\begin{array}}
\def\rig]{\end{array}\right]}
\def\begcen{\begin{center}}
\def\endcen{\end{center}}
\def\begrem{\begin{remark}\rm}
\def\endrem{\end{remark}}
\def\begdef{\begin{definition}}
\def\enddef{\end{definition}}
\def\begpro{\begin{proposition}}
\def\endpro{\end{proposition}}
\def\begfac{\begin{fact}}
\def\endfac{\end{fact}}
\def\begass{\begin{assumption}}
\def\endass{\end{assumption}}
\def\begsubequ{\begin{subequations}}
\def\endsubequ{\end{subequations}}

\def\begmat#1{\begin{bmatrix}#1\end{bmatrix}}

\def\cale{{\cal E}}

\def\calo{{\cal O}}
\def\calh{{\cal H}}

\def\cale{{\cal E}}

\def\phil{\phi^{\mathtt{L}}}

\def\bbm{{\mathbb{M}}}

\def\L2e{{\cal L}_{2e}}

\def\rea{\mathbb{R}}

\def\diag{\mbox{diag}}

\def\col{\mbox{col}}
\def\hal{{1 \over 2}}

\def\diag{\mbox{diag}}
\def\rank{\mbox{rank}\;}

\def\min{{\mbox{min}}}

\usepackage{color}

\def\QED{\hfill $\square$}
\def\qed{\hfill $\triangleleft$}

\usepackage[prependcaption,colorinlistoftodos]{todonotes}

\usepackage{tikz,xcolor,hyperref}

\definecolor{lime}{HTML}{A6CE39}
\DeclareRobustCommand{\orcidicon}{
	\begin{tikzpicture}
	\draw[lime, fill=lime] (0,0) 
	circle [radius=0.16] 
	node[white] {{\fontfamily{qag}\selectfont \tiny ID}};
	\draw[white, fill=white] (-0.0625,0.095) 
	circle [radius=0.007];
	\end{tikzpicture}
	\hspace{-2mm}
}

\foreach \x in {A, ..., Z}{\expandafter\xdef\csname orcid\x\endcsname{\noexpand\href{https://orcid.org/\csname orcidauthor\x\endcsname}
			{\noexpand\orcidicon}}
}

\input epsf

\def\BibTeX{{\rm B\kern-.05em{\sc i\kern-.025em b}\kern-.08em
    T\kern-.1667em\lower.7ex\hbox{E}\kern-.125emX}}
\begin{document}

\title{On the Equivalence of Contraction and Koopman Approaches for Nonlinear Stability and Control}
\author{\textsf{Bowen Yi\orcidA{}, and Ian R. Manchester\orcidB{}} 
%
\thanks{This work was supported by Australian Research Council under Grant DP190102963. \emph{(Corresponding author: Bowen Yi)}}

\thanks{B. Yi is with Department of Electrical Engineering, Polytechnique Montreal, Montreal, QC, Canada. The work has been done when he was with Australian Centre for Robotics, The University of Sydney, Australia. (email: b.yi@outlook.com)}
\thanks{I. R. Manchester is with Australian Centre for Robotics and School of Aerospace, Mechanical and Mechatronic Engineering, The University of Sydney, Sydney, NSW 2006, Australia. (email: ian.manchester@sydney.edu.au)
}
}

\maketitle

\begin{abstract}
In this paper we prove new connections between two frameworks for analysis and control of nonlinear systems: the Koopman operator framework and contraction analysis. Each method, in different ways, provides exact and global analyses of nonlinear systems by way of linear systems theory. The main results of this paper show equivalence between contraction and Koopman approaches for a wide class of stability analysis and control design problems. In particular: stability or stablizability in the Koopman framework implies the existence of a contraction metric (resp. control contraction metric) for the nonlinear system. Further in certain cases the converse holds: contraction implies the existence of a set of observables with which stability can be verified via the Koopman framework. We provide results for the cases of autonomous and time-varying systems, as well as orbital stability of limit cycles. Furthermore, the converse claims are based on a novel relation between the Koopman method and construction of a Kazantzis-Kravaris-Luenberger observer. We also provide a byproduct of the main results, that is, a new  method to learn contraction metrics from trajectory data via linear system identification.
\end{abstract}


%
\section{Introduction}
\label{sec1}
%
Learning, analysis, and control of nonlinear dynamical systems are important problems but are significantly more challenging than their linear counterparts. In recent years two approaches -- based on the Koopman operator and contraction analysis -- have become popular and facilitated significant progress. Each, in different ways, draws on linear systems theory to analyze nonlinear systems \textit{exactly} and \textit{globally}: the Koopman approach works by mapping the system state to high (possibly infinite) dimensional spaces of \textit{observables} in which the dynamics are linear, whereas the contraction framework analyzes the system via an infinite family of local linearizations. The main purpose of this paper is to make precise the connection between these two approaches.

The Koopman operator describes the dynamics over the space of observables. Even for nonlinear dynamics, the Koopman operator itself is always linear, and this fact can be utilized for many tasks including, {\em e.g.}, stability analysis, controller design, and system identification; see \cite{MAUetal} for a recent review.

Unlike linearization via first-order approximation, the Koopman method provides a global insight to the system behavior via spectral analysis of the operator, which is built upon the pioneering work \cite{KOO}. Furthermore, for a given system if we are able to find a finite set of Koopman eigenvalues and eigenfunctions, then the system can be immersed into a finite-dimensional linear dynamics. Indeed, the Koopman method is closely connected to the Hartman-Grobman theorem in the dynamical systems theory \cite{HAR} and nonlinear manifold learning in the machine learning community \cite{LINZHA}, which, however, are studied in a local manner. In \cite{MAUMEZ}, the Koopman operator is utilized to provide novel global stability criteria for both hyperbolic equilibria and limit cycles. It was shown that, with a distinct-eigenvalue assumption, the proposed criterion is necessary and sufficient for the asymptotic stability of an equilibrium. However, for many systems the Koopman operator can admit repeated eigenvalues and we are still able to map the dynamics into a linear system globally \cite{LANMEZ,KVAREV}. Importantly, an attractive feature in Koopman representations is the possibility of data-driven and model-free analysis \cite{MEZBAN,MEZnd}. To this end, several techniques have been proposed recently to approximate the Koopman operator by solving least square problems, {\em e.g.}, dynamic mode decomposition (DMD) \cite{SCHSES} and extended (E)DMD \cite{KORMEZ}.

Another way to study nonlinear systems by means of linear methods is \textit{contraction analysis}. A contracting system has the property that all its trajectories converge to each other, and contraction analysis is based on the study of a nonlinear system in terms of its \emph{differential dynamics}, which is a linear time-varying (LTV) system along solutions, in this way with linear systems theory being applicable. The key insight is, roughly speaking, that if all solutions are locally stable then all solutions are globally stable. The basic ideas can be traced back to \cite{LEW}, but remarkable utility for problems in control and observer design was noticed much later \cite{LOHSLO}, while connections to Lyapunov methods were elaborated in \cite{FORSEP}. Contraction is not just a method for system analysis, but also a powerful constructive tool in many control and learning tasks, see \cite{MANSLO,TOBetal,SINetal,YIetal} for recent applications in controller design, system identification, online learning and observer design. 

In contrast to Koopman operator, which is a single infinite-dimensional linear operator, contraction is based on analysis of an infinite family (along all feasible solutions) of finite-dimensional LTV systems. Due to the similarity -- using linear systems theory to analyze nonlinear systems globally -- between Koopman and contraction methods, one of the motivations of the paper is to clarify their connections. Despite existing research on Koopman representations for systems with an asymptotically stable equilibrium, it is the authors' opinion that contraction provides a more natural and fundamental form of nonlinear stability to link to Koopman representations than asymptotic stability. The main reason is that, for linear systems, asymptotic stability is a quite strong property, equivalent to various types of stability, {\em e.g.}, exponential, incremental, and input-to-state stability; contraction shares these strong properties, whereas asymptotic stability of nonlinear systems is relatively weak and sometimes fails to inherit these properties from the lifted linear dynamics in the Koopman framework.

The main contributions of this paper are:%
\begin{itemize}
\item[1)] Showing the equivalence between contraction and Koopman approaches for stability analysis of nonlinear systems. The stability conditions for autonomous and time-varying nonlinear systems in the Koopman framework imply contraction of the systems; and the converse results also hold true. We give the results for both equilibria and limit cycles with explicit constructions.

\item[2)] Establishing the links between control contraction metric (CCM) -- a concept to characterize \emph{universal stabilizability} of nonlinear systems \cite{MANSLO} -- and stabilizability of the lifted linear systems in the Koopman framework.

\item[3)] The constructive solutions in our converse results are obtained from some novel interesting links between the Koopman operator method and Kazantzis-Kravaris-Luenberger (KKL) observers \cite{BER}. The relevant results are of interest on their own.

\item[4)] An useful byproduct of the proposed equivalence between Koopman and contraction approaches is that we may learn contraction metrics for stable nonlinear systems from pure trajectory data.
\end{itemize}

{\em Notation}. $|\cdot|$ represents the Euclidean norm. All mappings and functions are assumed smooth. Given a matrix $M(x,t)$ and a vector field $f(x)$ with proper dimensions, we define the directional derivative as $\partial_f M(x,t) := \sum_i {\partial M(x,t) \over \partial x_i}f_i(x)$ and the operator $\nabla f := ({\partial f \over \partial x})^\top$. Given a point $y \in \rea^n$, we define a ball $B_{\varepsilon}(y):=\{x\in \rea^n | |x-y|<\varepsilon\} $. For square matrices $A$ and $B$, the notion $A \prec B$ (or $A\preceq B$) indicates $(A-B)$ negative definite (or semidefinite). We use $\rea^{m\times m}_{\succ 0}$ (or $\rea^{m\times m}_{\succeq 0}$) to represent the set of $m\times m$ positive definite (or semidefinite) matrices. A positive definite matrix-valued function $A:\rea^n \to \rea^{m\times m}_{\succ 0}$ is called uniformly bounded if $a_1 I \preceq A(x) \preceq a_2 I, \; \forall x$ with some $a_2\ge a_1>0$. We use $\mathbb{C}$ to represent the complex plane, and $\mathbb{C}_{>0}$ ($\mathbb{C}_{<0}$) for the open right (left) half-plane. For an open set, $\mathtt{cl}(\cdot)$ denotes its closure. When clear from the context, the arguments of mappings and operators are omitted. 

The paper is organized as follows. In Section \ref{sec2} we recall some preliminaries on the Koopman operator and contraction analysis. Section \ref{sec3} starts with the equivalence between two methods for nonlinear autonomous systems, and the main results are extended to time-varying systems in Section \ref{sec4}. Then, it is followed by the relevant results of limit cycles in Section \ref{sec5}. Section \ref{sec6} presents some discussions including the relation between the CCM and the Koopman conditions for stabilization problems, as well as the link between the Koopman method and construction of a KKL observer. Three examples are given in Section \ref{sec7} to illustrate the main results, and the paper is wrapped up with some concluding remarks in Section \ref{sec8}. A preliminary version of this paper was presented at the conference \cite{YIetal2}, containing the main results of Section \ref{sec3}.

\section{Preliminaries}
\label{sec2}

We consider both nonlinear autonomous systems of the form
\begin{equation}
\label{NL:auto}
 \dot x = f(x)
\end{equation}
and more general nonlinear time-varying (NLTV) systems
\begin{equation}
\label{NLTV}
\dot x = f(x,t),
\end{equation}
with state $x\in \rea^n$ and the vector field $f\in C^2$. The system is assumed complete, {\em i.e.}, having a unique solution $X(x,t)$ for $t\in [0,+\infty)$ from the initial condition $x$ at $t=0$ for the autonomous system \eqref{NL:auto}, or a solution $X(t;x,s)$ for the time-varying system \eqref{NLTV}. In the latter, $X(t;x,s)$ represents the solution value at time $t$ from the initial condition $x(s)$ at time $s$. In the remainder of the paper, we denote the Jacobian of the function $f$ as $F(x) = {\partial f \over \partial x}(x)$, or $F(x,t) = {\partial f \over \partial x}(x,t)$ for the time-varying case.


\subsection{Contraction Analysis and Control Contraction Metric}
\label{sec22}

Let us recall the definition of contraction \cite{LOHSLO}, which is given here in its global forms for simplicity.

\begin{definition}\rm \label{def:contraction}
(\emph{contraction}) Given a system \eqref{NLTV}, if there exists a uniformly bounded metric $M(x,t) \in \rea^{n\times n}_{\succ 0}$ such that
\begin{equation}
\label{contraction}
\dot{M} + {\partial f \over \partial x}(x,t)^\top M + M{\partial f \over \partial x}(x,t) \preceq - \rho M,
 \quad \rho >0,
\end{equation}
then we call the system \eqref{NLTV} contracting, and $M(x,t)$ is a contraction metric. If the right-hand side of \eqref{contraction} is replaced by ``$\prec 0$'', the given system is asymptotically contracting. \qed
\end{definition}

Contraction, also known as incremental exponential stability,  utilizes the differential behavior of an infinitesimal displacement $\delta x$ along the flow $X(x,t)$ to characterize the asymptotic behavior among trajectories of the system \eqref{NL:auto}. To be precise, contraction implies that any two trajectories will exponentially converge to each other, {\em i.e.}, $ \forall (x_a,x_b) \in \rea^n \times \rea^n$
$$
|X(x_a,t) - X(x_b,t)|\le k_0|x_a- x_b| e^{-\rho t},
$$
for some $k_0>0$. For the asymptotic case, it becomes
$$
|X(x_a,t) - X(x_b,t)|\le \beta(|x_a- x_b|,t),
$$
in which $\beta(\cdot,\cdot)$ is a class $\mathcal{KL}$ function.

In \cite{MANSLO2018,MANSLO}, contraction analysis was extended to the problem of nonlinear control as a {\em constructive} method, a benefit of which is its convex representation enabling computational methods for synthesis. The key there is the control contraction metric (CCM), which is related to \emph{universal stabilizability} -- a property characterizing if one can find a feedback law to track any feasible trajectories asymptotically. We discuss the problem of stabilization in Section \ref{sec51}.

\subsection{Koopman Operator Methods}

The Koopman operator was originally studied for the autonomous system \eqref{NL:auto}. Denote $\calo$ as the space of all $C^2$-smooth output functions $\varphi: \rea^n \to \mathbb{C}$, usually called observables.

\begin{definition}\rm 
\label{def:koopman}
(\emph{Koopman operator}) The Koopman semi-group of operators $U^t: \calo \to \calo$ associated to the flow $X(x,t)$ for the system \eqref{NL:auto} is defined by
$$
U^t [\varphi] = \varphi \circ X(x,t), \quad \varphi \in \calo.
$$
For an observable $\phi_\lambda \in \calo \setminus\{0\}$, the associated Koopman eigenvalue is defined as the constant $\lambda \in \mathbb{C}$ satisfying
\begin{equation}
\label{K_eig}
U^t  [\phi_\lambda] = e^{\lambda t}\phi_\lambda,
\end{equation}
if it exists, and we call $\phi_\lambda$ as a Koopman eigenfunction.
\qed
\end{definition}

The operator $U^t[\phi]$ has, at a point $x_A$, the value which $\phi$ has at the point $X(x_A,t)$ into which $x_A$ flows after the lapse of the time $t$ \cite{KOO}; see Fig. \ref{fig:1}. It follows immediately from its definition that the Koopman operator $U^t[\cdot]$ is linear. 

\begin{figure}[htp!]
    \centering
    \includegraphics[width= 0.95\linewidth]{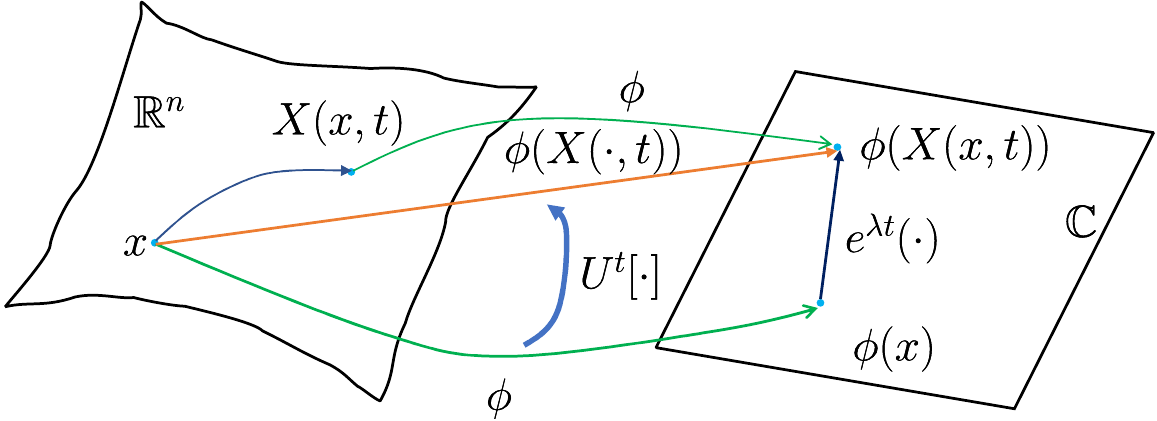}
    \caption{An illustration of the Koopman operator acting on an observable $\phi$.}
    \label{fig:1}
\end{figure}

Assuming smoothness of observables, the condition \eqref{K_eig} may be equivalently formulated as \cite{MAUetal}
\begin{equation}
\label{koopman_eig}
{\partial \phi_\lambda \over \partial x}(x) f(x)
 = \lambda \phi_\lambda.
\end{equation}
A main result in \cite{MAUMEZ} is that the Koopman operator can be used to verify the global asymptotic stability (GAS) of a hyperbolic equilibrium.
\begin{lemma}\rm
\label{lem:koopman} \cite[Proposition 2]{MAUMEZ}
 Consider the system \eqref{NL:auto} having a hyperbolic equilibrium $x_\star$, with all the eigenvalues of $F(x_\star)$ in $\mathbb{C}_{<0}$. If there exists a mapping 
$
\phi(x):=[\phi_{\lambda_1}(x), \ldots, \phi_{\lambda_n}(x)]^\top
$
such that\footnote{The following conditions are also necessary for the GAS of $x_\star$.}
\begin{enumerate}
\item ({\em distinct eigenvalues}) All the Koopman eigenvalues $\lambda_i$ are distinct and $\phi_{\lambda_i} \in C^1$ with $\nabla \phi_{\lambda_i}(x_\star) \neq 0$ for $i=1,\ldots, n$;
\item ({\em stability}) $\lambda_i \in \mathbb{C}_{<0}$ satisfy the PDEs \eqref{koopman_eig}, and are the eigenvalues of $F(x_\star)$.
\end{enumerate}
Then, the equilibrium $x_\star$ is GAS. \qed
\end{lemma}

The PDEs in the condition 2) immerse the autonomous system \eqref{NL:auto} into the linear time-invariant (LTI) dynamics 
\begequ
\label{lti}
\dot z = \Lambda z , \qquad  z(0) = \phi(x(0))
\endequ
with $\Lambda = \diag(\lambda_1,\ldots, \lambda_n)$ via the change of coordinate $x \mapsto z = \phi(x)$. By the condition $\lambda_i \in \mathbb{C}_{<0}$, this linear system is globally exponentially stable (GES) at the origin.



\section{Main Results for Autonomous Systems}
\label{sec3}

In this section, we start with the time-invariant system \eqref{NL:auto}, showing the equivalence between the Koopman and contraction approaches.

\subsection{Koopman Implies Contraction for Stability of Autonomous Systems}
\label{sec31}

In Lemma \ref{lem:koopman}, the second assumption requires distinct eigenvalues, requiring the LTI system \eqref{lti} \emph{diagonalizable}, which is somewhat restrictive\footnote{This could be relaxed by considering ``generalized eigenfunctions'' in \cite{MEZ17}.}. Before presenting our first result, we slightly extend Lemma \ref{lem:koopman} as follows.

\begin{proposition}\rm
\label{prop:koopman}
 Consider the system \eqref{NL:auto} with a hyperbolic equilibrium $x_\star$. If there exists a mapping 
$
\phi(x):=[\phi_{1}(x), \ldots, \phi_{N}(x)]^\top
$
with $(N-n) \in \mathbb{N}~\cup \{\infty\}$ such that
\begin{itemize}
\item[\bf C1] ({\em immersion}) For a finite $N$, $\Phi(x):={\partial \phi \over\partial x}(x)$ is full column rank; if $N$ is infinite, $\phi(x)$ is assumed to be rank-$n$ countably-infinite, {\em i.e.}, there are $\ell \ge n$ elements (denoted as $\phi_{k_j}$, $j=1,\ldots, \ell$) of $\phi$ such that $\rank\{\nabla \tilde \phi(x)\} = n$, with $\tilde\phi := \col(\phi_{k_1},\ldots, \phi_{k_\ell})$. 
\item[\bf C2] ({\em stability}) There exists exponentially stable $A$ verifying
\begin{equation}
\label{pde:new}
 {\partial \phi \over \partial x}(x) f(x) = A \phi(x). 
\end{equation}
\end{itemize} 
Then, the equilibrium $x_\star$ is GAS. 
\end{proposition}
\begin{proof}
We write the dynamics in the $z$-coordinate as
\begequ
\label{dot_z}
\dot z = Az, \quad z(0) = \phi(x(0)).
\endequ
From the assumption {\bf C2}, $z \to 0$ exponentially as $t\to\infty$. Since $x_\star$ is an equilibrium, we have $f(x_\star) =0$, and invoking \eqref{pde:new} yields $\phi(x_\star) = 0$. From the immersion condition {\bf C1}, if $N\in \mathbb{N}$, then $\phi: \rea^n \to \rea^N$ is \emph{locally} injective around $x_\star$, thus there is a class $\mathcal{K}$ function $\beta$ such that
$$
|x_a - x_b| \le \beta(|\phi(x_a) - \phi(x_b)|), \; \forall(x_a,x_b) \in B_\varepsilon(x_\star)^2
$$
with $\varepsilon>0$ sufficiently small. By substituting $x_b$ as $x_\star$ and $x_a$ as $X(x,t)$ with $x\in B_\varepsilon(x_\star)$, we conclude that $x_\star$ is locally attractive for the system \eqref{NL:auto}. In terms of hyperbolicity of the equilibrium $x_\star$, the Jacobian $F(x_\star)$ is Hurwitz.

Since $F(x_\star)$ is Hurwitz and the condition {\bf C1} holds for all $x \in \rea^n$, as well as invoking the fact that basin of attraction is open \cite[Proposition 5.44]{SAS}, the preimage $\phi^{-1}(0)$ only contains a single isolated equilibrium. Hence, the function $\phi$ is injective on the entire basin of attraction of $x_\star$. Due to the exponential stability of the $z$-dynamics, there always exists a moment $t_\star$ such that $x(t) \in B_\varepsilon(x_\star)$ for $t\ge t_\star$ from any initial condition. Then, for all $x(0) \in \rea^n$ we have
$$
\begin{aligned}
|x(t) - x_\star| & ~\le~ \beta(|e^{At}\phi(x(0)) - \phi(x_\star)|)\\
&
~\le~ \beta(|e^{At}\phi(x(0))|),
\quad \forall t\ge t_\star,
\end{aligned}
$$
which can be upper bounded by a $\cal KL$ function of $|x(0) - x_\star|$ and $t$, thus obtaining the GAS. If $N$ is infinite, we may use $\tilde \phi$ to do a similar analysis. 
\end{proof}

Note that the mapping $\phi$ in Proposition \ref{prop:koopman} is a little different from Koopman eigenfunctions, since it does not require the diagonalizability of $A$. Here we call $\phi$ the Koopman mapping, which is usually referred as semiconjugacy in the dynamical systems literature. The connection between Proposition \ref{prop:koopman} and Lemma \ref{lem:koopman} is established by the PDEs \eqref{koopman_eig} and \eqref{pde:new}. If the matrix $A$ is diagonalizable with $A = T^{-1} \Lambda T$ and $N \in \mathbb{N}$, then $T\phi$ is a set of Koopman eigenfunctions associated with the eigenvalues $\lambda_1, \ldots, \lambda_N$. Note that the mapping $\phi$ identified in {\bf C1}-{\bf C2} is not unique, in contrast to the principal eigenfunctions used in Lemma \ref{lem:koopman}.

\begin{remark}
\rm
The condition {\bf C1} implies the existence of a local inverse at each point, but not necessarily global, {\em i.e.}, it corresponds to an immersion rather than an embedding. The gap is that $\phi$ may be not proper in many cases due to $N \ge n$. In some cases we may want the stronger condition that $\phi$ has a left inverse $\phil$ with $\phil(\phi(x))=x$. Note that both trivially hold if $\phi$ contains the original system states, {\em i.e.}, $\phi(x)=\col(x, \phi'(x))$ for some $\phi'$. The existence of a Koopman mapping $\phi$ which satisfies the immersion condition is closely related to having a rectifiable dynamics \cite{KORMEZ2020,GOSPAL} such that it yields a local inverse.
\end{remark}

We are now in position to present the first main result of the paper.

\begin{theorem}
\label{thm:1}\rm
Assuming that there exists a $C^2$ Koopman mapping $\phi$ satisfying {\bf C1}-{\bf C2} and $\Phi^\top\Phi $ is uniformly bounded, then the system \eqref{NL:auto} is contracting with a contraction metric 
\begin{equation}
\label{M}
M(x) = \Phi(x)^\top P\Phi(x),
\end{equation}
where $P$ is the solution to $A^\top P + PA=-I$.\footnote{For infinite $N$, the Lyapunov equation becomes $\langle Az,Pz \rangle + \langle Pz, Az \rangle = -\langle z, z \rangle,\; z\in D(A)$ with $D(A)$ the domain of the infinitesimal generator $A$, and the inner product $
\langle \cdot, \cdot  \rangle
$ defined as \eqref{inner_product} for a compact set $x\in \mathcal{X}$.}
\end{theorem}
\begin{proof}
If there exists a Koopman mapping $\phi$ satisfying {\bf C1}-{\bf C2}, invoking that $A$ is exponentially stable, then there exists $P = P^\top \succ0$ satisfying the Lyapunov equation\footnote{This is also true for infinite-dimensional systems \cite[Thm. 5.1.3]{CURZWA}, and the following analysis is done \emph{mutatis mutandis} but omitted here.}
$$
A^\top P + P A = - I.
$$

From the assumptions, the mapping $\phi$ satisfies the PDE \eqref{pde:new}. Now, we calculate the partial derivative with respect to $x$ in each side, and denote $\Phi(x):= {\partial \phi \over \partial x}(x)$, yielding
\begequ
\label{dotPhi}
\begin{aligned}
	\partial_f \Phi(x) + {\partial \phi\over\partial x}(x){\partial f \over \partial x}(x) & ~=~ A\Phi(x)
	\\
	\implies  \hspace{1cm} \dot \Phi(x) + \Phi(x)F(x) &~ =~ A\Phi(x).
\end{aligned}
\endequ

From the Lyapunov equation, we have
\begequ
\label{phiphi}
\Phi^\top (A^\top P + P A) \Phi = - \Phi^\top \Phi \prec0,
\endequ
where we have used the full rank assumption of $\Phi(x)$ in {\bf C1}, and the fact that $\Phi(x)$ is a tall matrix. On the other hand,
$$
\begin{aligned}
	& \quad \Phi^\top (A^\top P + P A) \Phi \\
	\overset{\eqref{dotPhi}}=& \quad
	\Phi^\top P \dot \Phi + \dot\Phi^\top P \Phi + \Phi^\top P \Phi F + F^\top \Phi^\top P \Phi
	\\
	\overset{\eqref{M}}=& \quad 
	\dot M + F^\top M + M F
	\\
	\overset{\eqref{phiphi}}=& \quad 	-\Phi^\top\Phi.
\end{aligned}
$$
Applying
$$
\Phi^\top \Phi \succeq {1\over \lambda_{\tt max}\{P\}} \Phi^\top P\Phi
$$
with the largest eigenvalue $\lambda_{\tt max}\{P\}$, we then have
$$
\dot M + F^\top M + M F \prec - {1\over \lambda_{\tt max}\{P\}} M.
$$
It implies the contraction of the nonlinear system \eqref{NL:auto}.
\end{proof}


\begin{remark}\rm
The above proof boils down to the application of contraction of the lifted linear system $\dot z = Az $. Though it is well-known that incremental stability is intrinsic \cite{FORSEP}, the special point in the proof relies on the transformation $\phi: x\mapsto z$ being an immersion rather than a diffeomorphism. We underline that the immersion is guaranteed by the full rank condition {\bf C1}, which prevents $\phi$ from mapping another point $x'$ in a small neighborhood of the equilibrium $x_\star$ to the origin $z =0$ in the lifted coordinate. Furthermore, the full rank condition ensures that the contraction metric $M$ in \eqref{M} is positive definite; see Definition \ref{def:contraction}.
\end{remark}

\begin{remark}\rm 
It was shown in \cite{MAUMEZcdc} that the existence of a set of eigenfunctions $\phi_{\lambda_i}$ is related to a ``contracting metric'' 
\begequ
\label{mxy}
d(x_1,x_2) = \left(\sum_{i=1}^N |\phi_{\lambda_i}(x_1) - \phi_{\lambda_i}(x_2)|^p \right)^{1\over p}
\endequ
with integer $p\ge 1$, which follows the set stability framework to study incremental stability \cite{ANG}. In the past decade, there has been more attention on the \emph{differential framework} to analyze incremental stability, and our main results follow this line. It is clear that the obtained $M=\Phi^\top P \Phi$ is a Riemannian metric defined on tangent bundle, whereas \eqref{mxy} is defined in state space. The former enjoys attractive computational convenience in many settings \cite{MANSLO,MANSLOscl}.
\end{remark}

\begin{remark}\rm
The stability criterion in Proposition \ref{prop:koopman} requires finding a function $\phi$ and a matrix $A$ satisfying {\bf C1}-{\bf C2} simultaneously, making it non-trivial to verify. In \cite{MAUMEZ}, the authors provide numerical methods for the criterion in Lemma \ref{lem:koopman} in terms of Taylor expansion or Bernstein polynomials. 
\end{remark}

\subsection{A General Converse Result for Autonomous Systems}

In this subsection, we study the general converse results for nonlinear contracting autonomous systems, that is, the stability of infinitesimal generator of the Koopman operator $U^{t}$ for a class of observables.

To facilitate the general converse result, we consider in this section real-valued observables and the $L^2$ space, which is a Hilbert space with the usual inner product 
\begin{equation}
\label{inner_product}
\langle \phi_1, \phi_2  \rangle = \int_{x\in \mathcal{X}} \phi_1(x) \phi_2(x) dx,
\end{equation}
on a compact set $\mathcal{X}\subset \rea^n$. We have the following.

\begin{proposition}
\label{converse:general}\rm 
Suppose the system \eqref{NL:auto} is contracting with bounded trajectories in a compact set $\mathcal{X} \subset \rea^n$, and consider the set $\cale$ of real-valued observables $\phi$ parameterized as
\begin{equation}
\label{obs:prmt}
\phi(x) = \varphi(x) - \lim_{t\to\infty}\varphi(X(x_a,t))
\end{equation}
with any $C^2$-smooth function $\varphi$ and any point $x_a \in \mathcal{X}$. Then the infinitesimal generator of $U^{t}$, {\em i.e.}, $A_U \phi = \lim_{t\to 0}(U^{t}[\phi] -\phi)/t$ defines the linear infinite-dimensional system:
\begin{equation}\label{linsys_inf}
     {\partial \over \partial t} g(x,t) =  A_U g(x,t), \quad g(x,0) = \phi(x)
\end{equation}
and there exists a positive operator $P $ verifying its Lyapunov equation for exponential stability:
\begequ
\label{lyaeq_inf}
\langle A_U z,Pz \rangle + \langle Pz, A_U z \rangle = -\langle z, z \rangle
\endequ
for $z\in \cale$.
\end{proposition}

\begin{proof}
We first derive the representation \eqref{linsys_inf}. For the system \eqref{NL:auto}, we denote the output of the Koopman operator as\footnote{Here, $g(x,t)$ is thought as a function of \emph{isolated} coordinates $x$ and $t$ that evolve independently, see for example \cite[Remark 7.5.2]{LASMAC}.} 
$
g(x,t) := U^{t}[\phi(x)]
$
for given $t\ge  0$, with the initial condition constraint $g(x,0)= \phi(x)$. Invoking the smoothness assumption and according to \cite[Thm. 7.5.1]{LASMAC}, the infinitesimal generator is equivalently defined as
$$
\begin{aligned}
 {\partial \over \partial t} g(x,t)
 & = ~
  {\partial \over \partial t} U^{t}[\phi(x)]
=  ~
 {\partial \over \partial t} \phi(X(x,t))
 \\
 & = ~A_U g(x,t).
\end{aligned}
$$
Note that for any fixed $t\ge 0$, we have $g(x,t)\in \cale$.

Next we verify that $U^t$ is a strongly-continuous semigroup on the Hilbert space $L^2$.
This is a well-known property of the Koopman operator and indeed follows directly from  \cite[Def. 2.1.2]{CURZWA} and the continuity of $X(x,t)$ with respect to $t$.


We then note that any $C^2$ function on a compact set $\mathcal X$ is square-integrable, and hence is an element of the Hilbert space $L^2$. 
By assumption the system is contracting, and therefore for any $x$ in the compact set $\mathcal X$ we have $\|g(x,t)\| \le k e^{-\alpha t}\|g(x,0)\|$ for some $k, \alpha>0$, in which $\|\cdot\|$ is the function norm induced by the inner product $\langle \cdot, \cdot \rangle$ in the $L^2$ space. Therefore
\begin{equation}\label{square_int}
    \int_0^\infty \|U^tg(x,0)\|^2dt <\infty
\end{equation}
for any $g(x,0) \in \cale$.
By \cite[Lemma 5.1.2]{CURZWA} and \eqref{square_int}, we obtain the exponential stability of the operator $U^t$ which has the domain $\cale$:
\begin{equation}
    \|U^t\| \le k' e^{-\alpha' t}
\end{equation}
where $\|U^t\|$ denotes the operator induced norm on $L^2$, with some $k',\alpha'>0$. Therefore we have shown that the Koopman operator $U^t$ is a strongly-continuous and exponentially-stable semigroup defined on the observables set $\cale$. Following the same construction in \cite[Thm. 4.1.23]{CURZWA}, the observability gramian of the system $(A_U, I)$ defined for $z \in \cale$ by $P z = \int_0^\infty U^{t*} U^t z dt$ is well-posed due to \eqref{square_int}\footnote{$U^{t*}$ represents the adjoint operator of $U^t$; see \cite[Def. A.3.63]{CURZWA}.}, and the operator $P$ is a feasible solution to the Lyapunov equation \eqref{lyaeq_inf}.\footnote{Different from \cite[Thm. 4.1.23]{CURZWA}, there is no uniqueness guarantee for the operator $P$ in our case, since we do not study the completeness of $\cale$.}
%
\end{proof}

In the above we consider the infinite-dimensional case with $A_U$ a stable operator. This general result, however, is difficult to use in practice. We are more interested in a finite-dimensional $A_U$. Such a problem will be studied in the next subsection for autonomous systems.

\subsection{Contraction Implies Koopman for Stability of Autonomous Systems}
\label{sec32}

In this subsection, we prove the converse of Theorem \ref{thm:1}, {\em i.e.}, contraction is also sufficient for {\bf C1}-{\bf C2}. As a result, we may use the convex condition \eqref{contraction} to verify the conditions {\bf C1}-{\bf C2} indirectly, {\em e.g.}, by means of convex optimization and sum-of-squares \cite{MANSLO,MANSLOscl}. We have the following.

\begin{theorem}\rm 
\label{thm:inv1}
Consider the autonomous system \eqref{NL:auto}, which is contracting with the metric $M(x)$ in a compact set $\mathtt{cl}(\mathcal{X})$. Then, there always exists a $C^1$ Koopman mapping $\phi$ satisfying {\bf C1}-{\bf C2} with $N\in \mathbb{N}$.
\end{theorem}

\begin{proof}
For a system forward invariant in a closed set, if it is \emph{autonomous} and contracting, we can prove that the system admits a unique GES equilibrium $x_\star$, {\em i.e.} $f(x_\star) = 0$, using the well-known Banach fixed point theorem; see for example \cite{HEEetal}. (It holds true for the case that the set $\cal X$ is exactly $\rea^n$, which is both closed and open.) Since we can always assign the equilibrium by a linear coordinate change $x \mapsto (x - x_\star)$, without loss of generality, we assume the equilibrium $x_\star$ at the origin.

The system dynamics \eqref{NL:auto} can be written as
$$
\dot x = F_\star x - H(x), \quad F_\star := F(x_\star)
$$
with the high-order remainder term 
\begin{equation}
\label{H}
H(x) := -f(x) + F_\star x.
\end{equation}
From the contraction assumption, we have
$$
\partial_f M(x) + F(x)^\top M(x) + M(x)F(x) \le -\gamma M(x)
$$
with $\gamma>0$ along all feasible solutions, and the metric $M(x)$ uniformly bounded. After substituting the particular solution
$
X(x_\star,t)= x_\star, ~ \forall t\ge 0,
$
we obtain
$$
F_\star^\top M(x_\star) + M(x_\star) F_\star \le -\gamma M(x_\star),
$$
in which we have used $f(x_\star) = 0$, thus $\dot M(X(x_\star,t)) =0$. The above inequality implies that $F_\star$ is Hurwitz. Hence, $x_\star$ is a hyperbolic equilibrium. 

We consider the case $N=n$, and parameterize $\phi$ as 
$$
\phi(x)= x + T(x)
$$ 
with a smooth function $T$ to be searched for. Substituting \eqref{H} into the PDE \eqref{pde:new} and fixing $A=F_\star $, we have
\begin{equation}
\label{pde:kkl}
\begin{aligned}
 &	\quad \left[ I + {\partial T\over \partial x}(x) \right] f(x) = A(x + T(x))
 \\
 \implies & \quad F_\star x - H(x) + {\partial T\over \partial x}(x)f(x) = AT(x) + Ax
 \\
 \overset{A=F_\star}{\implies} & \quad
 {\partial T\over \partial x}(x) f(x) = AT(x) + H(x) .
\end{aligned}
\end{equation}
Note that we have already shown $A=F_\star$ is Hurwitz from the contraction assumption. It is interesting to observe that the last PDE in \eqref{pde:kkl} is identical to the one appearing in the KKL observer \cite{KAZKRA}. That is, finding a mapping $\phi$ verifying {\bf C2} is equivalent to solving the PDE in KKL observers with respect to the new mapping $T$.\footnote{We refer the interested reader to Section \ref{sec62} for a brief introduction to KKL observers.}

The remainder of the proof follows some constructive solutions to KKL observers \cite{KAZKRA}, in which it is assumed that the solution $X(x,t)$ does not blow-up in finite \emph{backward} time. When this additional assumption does not hold, we may still continue the analysis by considering the modified dynamics
\begin{equation}
\label{modified}
\dot x = \rho(x) f(x)
\end{equation}
with an arbitrary $C^\infty$ function $\rho: \rea^n \to\rea$ satisfying
$$
\rho(x) = \left\{
\begin{aligned}
1, \quad &\mbox{if~~} x\in \mathtt{cl}(\mathcal{X})
\\
0, \quad & \mbox{if ~~} x\notin \mathcal{X}'
\end{aligned}
\right.
$$
for some $\mathtt{cl}(\mathcal{X}) \subset \mathcal{X}' \subset \rea^n$. Such a modification makes the above backward assumption always hold. Note that we may select the compact set ${\tt cl}(\cal X)$ arbitrarily large. Let us denote the solution of the modified dynamics \eqref{modified} as $\breve{X}(x,t)$.

As shown in \cite{ANDPRA}, the PDE \eqref{pde:kkl} has a feasible solution
\begin{equation}
\label{solution:kkl}
T(X(x,t)) = e^{F_\star t} T(x) +\int_0^t e^{F_\star (t-s) } H(X(x,s))ds\
\end{equation}
if it is well-posed. Since $F_\star$ is Hurwitz, it has a solution when $t\to +\infty$, and we select
\begequ
\label{T}
T(x) = \int_0^{+\infty} \exp(F_\star s) H(\breve{ X }(x,-s))ds.
\endequ
Therefore, by selecting 
$$
\phi^0(x) = x + T(x)
$$ 
with the $C^1$ function $T$ defined in \eqref{T}, we get a feasible solution in $\mathtt{cl}(\mathcal{X})$ to the 
PDE \eqref{pde:new}. 

The remainder is to verify the full rank condition {\bf C1} by redesigning the function $\phi^0(x)$. Combining the facts 
$$
T(x_\star) = 0, \quad \nabla T(x_\star) =0
$$
and 
$$
{\partial \phi^0\over\partial x}(x) = I_n + {\partial T \over \partial x}(x),
$$
we conclude from the continuity that there always exists a small parameter $\varepsilon>0$ such that the mapping $\phi^0(x)$ is an \emph{injection} in the neighborhood $B_\varepsilon(x_\star)$ of the equilibrium $x_\star$. 
However, it may be not true for the entire $\rea^n$. Motivated by \cite[Thm. 2.3]{LANMEZ} and \cite{WAN}, we need to modify the mapping $\phi^0(x)$. Since we may, via the change of coordinate $z = \phi^0(x)$, transform the dynamics into
$
\dot z = F_\star z,
$
we now write its flow as $Z(z,t)= e^{At}z$ from the initial condition $z\in \rea^n$, and the dynamics is complete for $t\in (-\infty,+\infty)$. Note that, for the above construction, we have $Z(z,t) =  \phi^0(X(x,t))$ for any $t$ and $z=\phi^0(x)$. We may derive the solution of $Z(\cdot)$ at $t_x>0$ in the lifted linear $z$-coordinate as
\begin{equation}
\label{ztx}
    Z(z,t_x) = e^{At_x} z = e^{At_x} \phi^0(x).
\end{equation}
On the other hand, we may calculate its solution using the flow in the original $x$-coordinate, and then lifting it to the $z$-coordinate, {\em i.e.}
\begin{equation}
\label{ztx2}
Z(z,t_x) = \phi^0(X(x,t_x)).
\end{equation}
By combining \eqref{ztx} and \eqref{ztx2}, we have 
\begin{equation}
\label{ztx3}
\phi^0(x) = e^{-A t_x} \phi^0(X(x,t_x)).
\end{equation}
Hence, re-defining the Koopman mapping as
\begequ
\label{phi_redesign}
\phi(x)  :=  e^{-A t_x} \phi^0(X(x,t_x)),
\endequ
which still satisfies {\bf C2} for any $t_x>0$. Regarding {\bf C1}, we have 
$$
\Phi(x)
=
e^{-At_x}{\partial \phi^0 \over \partial x}(X(x,t_x)){\partial X \over \partial x}(x,t_x).
$$
For a very large $t_x>0$, $X(x,t_x)$ is in a small neighborhood of the origin, and thus both $\nabla \phi^0(X(x,t_x))$ and ${\partial X \over\partial x}(x,t_x)$ are full rank. Thus, the Jacobian $\Phi(x)$ is full rank globally.
\end{proof}

\begin{remark}
\rm 
It is shown in \cite{MAUMEZ} that the assumptions in Lemma \ref{lem:koopman} are also necessary for GAS of autonomous systems, similar to the converse result in Theorem \ref{thm:inv1}. However, the difference is clear that Lemma \ref{lem:koopman} requires a distinct-eigenvalue assumption, making it somewhat restrictive, and instead, we impose an immersion condition. As a consequence, the PDEs in Lemma \ref{lem:koopman} and {\bf C2} are slightly different depending on whether the system matrix $A$ is diagonalizable.
\end{remark}

\begin{remark}
In the above proof, we verify the condition {\bf C2} by studying the flows in both the $x$- and $z$-coordinates. Indeed, we may verify directly the PDE \eqref{pde:new} for the re-designed mapping $\phi$ in \eqref{phi_redesign}. To be precise, we have
$$
\begin{aligned}
{\partial \phi \over \partial x}(x)f(x) 
& =
e^{-A t_x}{\partial \phi^0 \over \partial x}(X(x,t_x)) {d X\over d t}(x,t_x)
\\
& =
e^{-A t_x}{\partial \phi^0 \over \partial x}(X(x,t_x)) f(X(x,t_x))
\\
& = 
e^{-A t_x} A \phi^0(X(x,t_x))
\\
& = 
 A e^{-A t_x} \phi^0(X(x,t_x))
\\
& = A \phi(x).
\end{aligned}
$$
\end{remark}

\begin{remark}\rm 
In both Theorem \ref{thm:1} and its converse claim -- Theorem \ref{thm:inv1} -- we consider the stronger condition of full rank of $\Phi(x)$ than those used in the previous results, showing the equivalence of the Koopman approach to contraction analysis, not just the stability of the origin. This fact will be further elaborated for time-varying systems.
\end{remark}

\begin{remark}\rm\label{remark:C}
It should be underscored that there is a gap on smoothness between Theorem \ref{thm:1} and its converse result. The standing assumption in Theorem \ref{thm:1} is $\phi\in C^2$; however, in the converse result in Theorem \ref{thm:inv1} our construction of $\phi$ in \eqref{phi_redesign} is only $C^1$. In \cite{KVAREV} it is shown that $\phi \in C^2$ defined on the entire domain of attraction exists by imposing the stronger assumption of $f\in C^\infty$ and the nonresonance of $F(x_\star)$. Note that the nonresonance condition is also used in \cite[Eq. (5)]{KAZKRA} to study the solution to the PDE \eqref{pde:kkl} in the context of KKL observers.
\end{remark}

\section{Main Results for Time-Varying Systems}
\label{sec4}

In this section, we extend the results in Section \ref{sec3} from autonomous systems to NLTV systems.

\subsection{The Koopman Operator for NLTV systems}
\label{sec41}
The Koopman operator was originally defined for autonomous systems, and its definition was extended to NLTV systems in \cite{MACetal} and controlled systems with inputs in \cite{PROetal}. Consider $\calo$ as the space of all $C^2$-smooth observables $\varphi: \rea^n \times \rea \to \mathbb{C}$. In consistent with \cite{MACetal}, we define the \emph{non-autonomous} Koopman operator for the NLTV system \eqref{NLTV} with two parameters $(s,t)$ as
\begequ
\label{Koopman_operator:TV}
U^{(t,s)}[\varphi(x,s)] := \varphi(X(t;x,s ),t),
\endequ
which is able to characterize the time-varying property of the system; see Fig. \ref{fig:2}. In the non-autonomous Koopman operator, we allow the observables being time-varying, which is consistent with Definition \ref{def:koopman}. Note that for such an extension, the Koopman operator is still linear.

\begin{figure}[h]
    \centering
    \includegraphics[width= 0.95\linewidth]{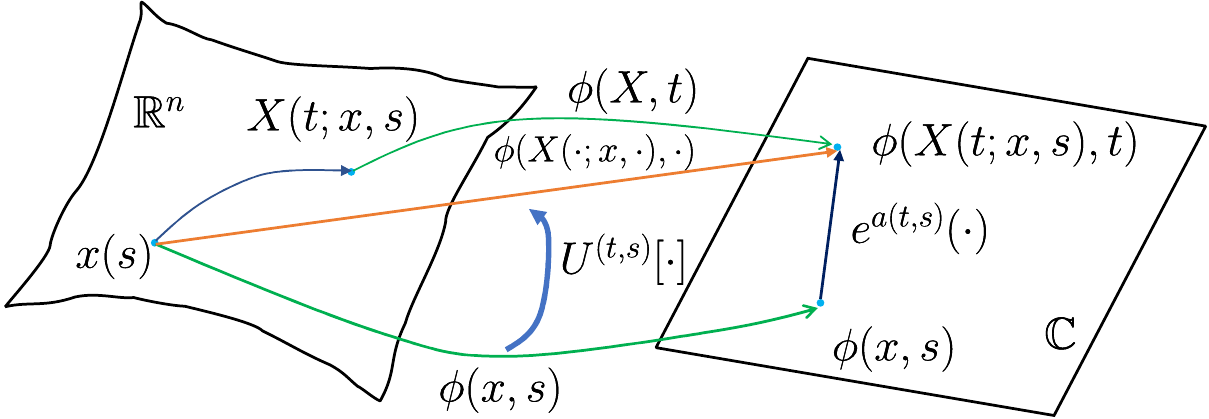}
    \caption{Non-autonomous Koopman operator acting on a time-varying observable $\phi(x,t)$}
    \label{fig:2}
\end{figure}

In a similar way, we define the Koopman eigenfunctions $\phi_\lambda(x,t)$ and Koopman eigenvalues $\lambda: \rea \to \mathbb{C}$ by
$$
a(t,s) = \int_s^t \lambda(\tau) d\tau,
$$
and then
\begequ
\label{eig_tv}
U^{(t,s)}[\phi_\lambda] = e^{a(t,s)} \phi_\lambda,
\endequ
which, for the autonomous systems, clearly degenerates into $U^{(t,s)}[\phi_\lambda] = e^{\lambda(t-s)}\phi_\lambda$. It is compatible with the autonomous definition selecting $s=0$.

The infinitesimal generator $A_U^s$ of the non-autonomous Koopman operator $U^{(t,s)}[\cdot]$ is defined as \cite{MACetal}
$$
\begin{aligned}
A_U^s \varphi(x,t) 
& ~:=~
\lim_{t \to s} {U^{(t,s)}[\varphi(x,s)] - \varphi(x,s) \over t -s }
\\
& ~\hspace{0.1cm}= ~
\lim_{t\to s}  {\varphi(X(t;x,s),t) - \varphi(x,s) \over t -s }
\\
& ~\hspace{0.1cm}=~
\dot{\aoverbrace[L1R]{\varphi(x,t)}}\Big|_{t=s},
\end{aligned}
$$
assuming that the observable $\varphi$ is smooth. Now considering the eigenfunction $\phi_\lambda(x,t)$ and invoking \eqref{eig_tv}, we have
$$
\begin{aligned}
\dot{\aoverbrace[L1R]{\phi_\lambda(x,t)}} 
&~ =~
 {\partial \phi_\lambda \over \partial t} +
 {\partial \phi_\lambda \over \partial x}f(x,t)
 \\
&~ =~
\lim_{t\to s} {e^{a(t,s)} - 1 \over t -s } \lim_{t\to s}\phi_\lambda(t,s)
\\
&~ = ~
{d a(s,t) \over dt}\Big|_{s=t} \phi_\lambda(x,t)
\\
&~ = ~
\lambda(t)\phi_\lambda(x,t).
\end{aligned}
$$
Equivalently, the Koopman eigenfunction $\phi_\lambda \in \calo $ and the eigenvalue $\lambda\in \mathbb{C}$ satisfies the PDE
\begin{equation}
\label{koopman_eig_TV}
 {\partial \phi_\lambda \over \partial t} + 
 {\partial \phi_\lambda \over \partial x   }f(x,t)
 = \lambda \phi_\lambda.
\end{equation}

If we are able to find a family of eigenfunctions $\phi(x,t):= [\phi_{\lambda_1}, \ldots, \phi_{\lambda_N}]^\top$ with $N \ge n$, then the given NLTV system \eqref{NLTV} can be immersed into an LTV system\footnote{It should be underlined that, in many cases, we are still able to find \emph{constant} Koopman eigenfunctions for some classes of time-varying systems \eqref{NLTV}. As a result, the resulting lifted linear systems are time-invariant. }
$$
\dot z = \Lambda(t) z
$$
with $z= \phi(x,t)$ and $\Lambda(t):= \diag(\lambda_1(t),\ldots, \lambda_N(t))$.

Following Subsection \ref{sec31}, we define the (time-varying) Koopman mapping $\phi$ to remove the diagonalizable constraint, which is required to satisfy the PDE
\begin{equation}
\label{pde2}
{\partial \phi \over \partial t}(x,t) + {\partial \phi \over \partial x}(x,t) f(x,t) = A(t)\phi(x,t)
\end{equation}
with a time-varying matrix $A(t)$ Lyapunov stable, {\em i.e.}, for any positive definite matrix $Q$ there exists a positive definite matrix $P(t)$ satisfying \cite{TOMBAN}
\begin{equation}
\label{Lya_eq_tv}
\dot P(t) + A^\top(t)P(t) + P(t) A(t)  + Q = 0.
\end{equation}
If $P(t)$ is uniformly bounded, the lifted LTV system
\begequ
\label{ltv_z}
\dot z = A(t) z, \quad z(0)= \phi(x(0),0)
\endequ
is uniformly exponentially stable (UES).

\subsection{Koopman Implies Contraction for NLTV Systems}
\label{sec42}

In this subsection, we show that a variant of Theorem \ref{thm:1} holds for time-varying systems. 

\begin{theorem}
\label{prop:tv}\rm
Consider the NLTV system \eqref{NLTV}, and assume the existence of a $C^2$-Koopman mapping $\phi(x,t) \in \rea^N$ satisfying 
\begin{itemize}
\item[\bf C1$'$] ({\em immersion}) For $N \in \mathbb{N}$,
\begequ
\label{rank_phi}
\rank \left\{{\partial \phi(x,t)\over \partial x} \right\} = n,
\quad \forall x\in \rea^n, \; t \in \rea_{\ge 0}.
\endequ
\item[\bf C2$'$] ({\em stability}) The existence of a UES and bounded $A(t)$ verifying the PDE \eqref{pde2}.
\end{itemize}
If $[{\partial \phi(x,t)\over \partial x} ]^\top {\partial \phi(x,t)\over \partial x}$ is uniformly bounded, then the system is contracting. 
\end{theorem}
\begin{proof}
We define the partial derivative ${\Phi}(x,t) := {\partial \phi \over \partial x}(x,t)$. For the PDE \eqref{pde2}, by calculating the partial derivative with respect to $x$, we obtain
$$
{\partial^2 \phi \over \partial x \partial t}  + \partial_f \Phi +{\partial \phi \over \partial x} {\partial f \over \partial x} = A{\partial \phi \over \partial x}.
$$
Since $\phi\in C^2$, we have
$
{\partial ^2 \phi \over \partial x\partial t}(x,t) = {\partial^2 \phi \over \partial t\partial x}(x,t)
$
then
$$
\begin{aligned}
 & \quad {\partial^2 \phi \over \partial t \partial x}  + \partial_f \Phi +{\partial \phi \over \partial x} {\partial f \over \partial x} =  A{\partial \phi \over \partial x}
 \\
 \iff & \quad
 {\partial \Phi \over \partial t} + \partial_f \Phi + \Phi F= A\Phi
 \\
 \iff & \quad   \dot{\aoverbrace[L1R]{\Phi(x,t)}} + \Phi(x,t)F(x,t) =  A(t) \Phi(x,t).
\end{aligned}
$$

From the UES of $A(t)$, there exists $P(t) \succ 0$ satisfying the differential Lyapunov inequality \cite[Thm. 7.8]{RUG}
\begin{equation}
\label{lya_ineq}
\dot P(t) + A^\top(t) P(t) + P(t) A(t) \preceq - kI_n
\end{equation}
and the uniform boundedness
$$
p_1 I_n \preceq P(t) \preceq p_2 I_n, \quad \forall t\ge 0
$$
for some $k>0$ and $p_2>p_1 >0$.

Choose the contraction metric 
$$
M(x,t) = \Phi(x,t)^\top  P(t) \Phi(x,t),
$$
in which $P(t)\succ 0$ is the solution to the differential Lyapunov inequality \eqref{lya_ineq}. Then, it yields
$$
\begin{aligned}
& ~\dot M + F^\top M + MF  
\\
= & ~
\dot\Phi^\top P \Phi + \Phi^\top P \dot \Phi + \Phi^\top \dot P \Phi  + F^\top \Phi^\top P \Phi + \Phi^\top P \Phi F
\\
= & ~
\Phi^\top(A^\top P + P A + \dot P ) \Phi
\\
\preceq & ~
- kI_n.
\end{aligned}
$$
Since both $P(t)$ and $\Phi^\top \Phi$ are uniformly bounded, the matrix $M(x,t)$ defined above is qualified as a uniformly bounded contraction metric. It completes the proof.
\end{proof}

The above result is not surprising, since contraction analysis was originally tailored for time-varying systems. From Theorems \ref{thm:1} and \ref{prop:tv}, the Koopman-based stability analysis may be roughly viewed as an alternative formulation of contraction analysis.

\subsection{Contraction Implies Koopman  for NLTV Systems}
\label{44}

In this subsection, we give two converse results on how to derive Koopman mappings for a contracting system. We first consider the case to embed a contracting NLTV system into an LTV system.

\begin{theorem}
\label{thm:conv_ltv}\rm 
Consider the NLTV system \eqref{NLTV} with an equilibrium $x_\star$ at the origin, and assume that the system is contracting with a smooth \emph{time-varying} metric $M(x,t)$ in a compact set $\mathtt{cl}(\mathcal{X})$. Then, there always exists a Koopman mapping $\phi(x,t)$ satisfying {\bf C1$'$} and {\bf C2$'$}. Further, the claims still hold true if the systems state is bounded but removing the assumption of the existence of an equilibrium $x_\star$.
\end{theorem}


\begin{proof}
From the contraction of the system \eqref{NLTV}, we know the UES of the LTV system
$$
\delta \dot x = {\partial f \over \partial x}(X(t;x,s),t) \delta x, \quad t\ge s
$$
with the infinitesimal displacement $\delta x \in T\rea^n$, any $s\ge 0$ and 
$$
\dot M + {\partial f \over \partial x}^\top M + M{\partial f \over \partial x} \preceq - \rho M, \quad \rho >0.
$$
As a particular solution from the equilibrium $x_\star$, we conclude the UES of 
$$
A(t) :={\partial f\over \partial x}(x_\star,t),
$$
{\em i.e.}, satisfying the differential Lyapunov inequality \eqref{lya_ineq} with
$$
P(t) = M(X(t;x_\star,0),t) = M(x_\star,t).
$$

Now let us parameterize the Koopman mapping as
$$
\phi(x,t) = x + T(x,t),
$$
with a $C^1$-mapping $T: \rea^n \to \rea_{\ge 0}$ to search, and define
$$
H(x,t) = - f(x,t) + A(t).
$$

Then, the PDE \eqref{pde2} in the non-autonomous Koopman approach becomes
\begin{equation}
\label{kklpde-tv}
{\partial T\over \partial x}(x,t) + {\partial T \over \partial x}(x,t) f(x,t) = A(t)T(x,t) + H(x,t),
\end{equation}
with a stable $A(t)$ from $s$ to $t$. We denote the associated \emph{state transition matrix} of $A(t)$ as $\Omega(t,s)$, and we have
$$
z(t) = \Omega(t,s) z(s) 
$$
for the LTV system \eqref{ltv_z}. It is a function of $A(t)$, but most time it is impossible to write its explicit formula, and we may numerically obtain from
$$
\dot \Phi_A(t)= A(t) \Phi_A(t), \quad \Phi_A(0) = I_n
$$
with $\Omega(t,s)= \Phi_A(t)\Phi_{A}(s)^{-1}$.

The next step is to show the mapping
\begin{equation}
\label{T0_tv}
T^0(x,t) = \int_0^t\Omega(t,s)H( X(s;x,t),s) ds
\end{equation}
is a feasible solution to the PDE \eqref{kklpde-tv} for $x\in \mathcal{X}$.\footnote{Here, we assume that for the NLTV system \eqref{NLTV} there is no backward finite-time escaping. Otherwise, we may adopt the modification in \eqref{modified}.} To this end, motivated by \cite{BERAND}, we note that for any $x\in \rea^n$, $t\in[0,+\infty)$ and any $\tau$, we have
$$
X(s;X(t+\tau,x,t),t+\tau) = X(s;x,t),
$$
and then
\begequ
\label{equ_tv}
\begin{aligned}
&~	T^0(X(t+\tau;x,t), t+\tau)
\\
 =&~
\int_0^{t+\tau}\Omega(t+\tau,s) H(X(s;x,t), s) ds
\\
 = & ~
\int_0^{t+\tau}\Omega(t+\tau,t) \Omega(t,s) H(X(s;x,t), s) ds
\\
 = &~ 
\Omega(t+\tau,t)  \int_0^{t+\tau}\Omega(t,s) H(X(s;x,t), s) ds
\\
 = & ~
\Omega(t+\tau,t)  T^0(x,t)
\\
& ~~ +
\Omega(t+\tau,t)  \int_t^{t+\tau}\Omega(t,s) H(X(s;x,t), s) ds.
\end{aligned}
\endequ
Invoking the $ C^1$-smoothness, we have
$$
\lim_{\tau \to 0} {T^0(X(t+\tau;x,t), t+\tau)\over \tau}  = {\partial T^0 \over \partial t}(x,t) + {\partial T^0 \over \partial x}(x,t)f(x,t)
$$
and
$$\small
\begin{aligned}
&\lim_{\tau \to 0} {\Omega(t+\tau,t)  T^0(x,t) \over \tau}  = {d \Omega \over dt}(t,s)\Big|_{s=t} T^0(x,t)=A(t)T^0(x,t)
\\
&\lim_{\tau \to 0} {1\over \tau} \Omega(t+\tau,t)  \int_t^{t+\tau}\Omega(t,s) H(X(s;x,t), s) ds  =
H(x,t).
\end{aligned}
$$
Combining \eqref{equ_tv}, we verify that $T^0(x,t)$ defined by \eqref{T0_tv} is a feasible solution to the PDE \eqref{kklpde-tv}, and then 
$$
\phi^0(x,t) = x+ T^0(x,t)
$$
is a $C^1$-solution to \eqref{pde2} in $\mathcal{X}$.

The next step is to verify the rank condition \eqref{rank_phi}. We have
$$
{\partial \phi^0 \over \partial x}(x,t) = I_n + {\partial T^0 \over \partial x}(x,t),
$$
with 
\begin{equation}
\label{partial_H2}
H(x_\star, t) = - f(x_\star,t) + A x_\star =0, \quad {\partial H \over \partial x}(x_\star,t)=0,
\end{equation}
and then $T(x_\star,t) = 0$ and $\nabla T(x_\star,t)=0$. It implies the existence of a small parameter $\varepsilon>0$ such that the mapping $\phi^0(x,t)$ is an injection (for any fixed $t$) in $B_\varepsilon(x_\star)$. Similarly to the proof of Theorem \ref{thm:inv1}, we consider the flow in the $z$-coordinate as
$$
z(t_x) = \Omega(t_x,t) z(t) = \Omega(t_x,t) \phi^0(X(t;x,s),t),
$$
as well as the flow in the $x$-coordinate
$$
z(t_x) = \phi^0(X(t_x;x,s),t_x).
$$
Hence, choosing $s=t$ and using $X(t;x,s)|_{s=t}=x$ and $\Omega(t_x,t)\Omega(t,t_x) = I$, we modify the transformation as
$$
\phi(x,t) = \Omega(t,t_x) \phi^0(X(t_x; x, t), t_x),
$$
which satisfies both the PDE \eqref{pde2} and the immersion condition \eqref{rank_phi} with sufficiently large $t_x>0$.

In the above proof, we have assumed the existence of an equilibrium $x_\star$. Now, let us consider the case without equilibria. From the assumption of contraction, we know that all trajectories will converge to each other ultimately. Selecting a particular solution $x_r(t)$, we have
\begequ
\label{xr}
|x(t) - x_r(t)|\le k_0 |x(0) - x_r(0)| e^{-\rho t}, \quad x(0)\in \rea^n 
\endequ
with some $k_0, \rho>0$. We define an error state $\chi:= x - x_r$, the dynamics of which is given by
\begequ
\label{dotchi}
\dot \chi = f_\chi(\chi,t)
\endequ
with
$$
f_\chi(\chi,t) = f(\chi+ x_r(t),t) - f(x_r(t),t).
$$
Clearly, this system is GES at the origin due to \eqref{xr}. Then, we may apply the result with an equilibrium to the dynamics \eqref{dotchi}. The associated Koopman transformation is given by $\phi(x-x_r(t),t)$, where $\phi(\cdot,t)$ follows the construction above.
\end{proof}

Sometimes there is an interesting special case to identify \emph{constant} Koopman eigenvalues, or equivalently to lift the nonlinear model into an LTI system. It is always possible to do that, but imposing the immersion condition {\bf C1} requires that the Jacobian of the vector field should not change significantly over time. We have the following.

\begin{proposition}\rm
\label{prop:inv2}
Consider the NLTV system \eqref{NLTV} with an equilibrium $x_\star$ at the origin, and assuming the system is contracting with a metric $M(x)$ in the set $\mathtt{cl}(\mathcal{X})$. There always exists a Koopman mapping $\phi(x,t)$ satisfying {\bf C1$'$} but with a constant Hurwitz matrix $A\in \rea^{n\times n}$. Further, if 
\begin{equation}
\label{jacobian_f}
\left|{\partial f\over \partial x}(0,0)- {\partial f \over \partial x}(0,t) \right| \le k, \quad \forall t\ge 0,
\end{equation}
for some small $k>0$, then the condition {\bf C2$'$} also holds.
\end{proposition}
\begin{proof}
From the contraction assumption, we have
$$
\dot M(x) + F^\top(x,t) M(x) + M(x) F(x,t) \le - \gamma M(x)
$$
for some $\gamma >0$, with $F(x,t)={\partial f(x,t) \over \partial x}$. Since $x_\star =0$ is an equilibrium, we have $\dot M(X(x_\star,t)) = 0$, thus
$$
F^\top (x_\star, t) M(x_\star) + M(x_\star) F(x_\star,t) \le - \gamma M(x_\star).
$$
It implies that the Jacobian $F(x_\star,t)$ is Hurwitz at any fixed time. Now we select the Hurwitz matrix
$
A : = F(x_\star, 0),
$
which may work in our analysis due to the assumption \eqref{jacobian_f}. Now we consider $N=n$ and parameterize $\phi$ as
\begequ
\label{parameterisation2}
\phi(x,t) = x + T(x,t)
\endequ
and define
$
H(x,t):= -f(x,t) + Ax.
$ 

The existence of a Koopman mapping $\phi$ is now equivalent to the existence of the solution to the PDE \eqref{pde2} which satisfies the immersion condition. With the above parameterization \eqref{parameterisation2}, the equation \eqref{pde2} becomes
\begin{equation}
\label{pde3}
{\partial T \over \partial t}(x,t) + {\partial T\over \partial x}(x,t)f(x,t) = A T(x,t) + H(x,t).
\end{equation}
It is interesting to figure out that the PDE \eqref{pde3} is exactly the same one in the KKL observer for \emph{non-autonomous} systems \cite{BER,BERAND}. Mimicking the operation done in \eqref{modified}, we may modify the NLTV system as $\dot x = \rho(x) f(x,t)$ with its solution denoted as $\breve{X}(s;x,t)$. That is, $\breve X(s;x,t)$ represents the solution value at time $s$ from the initial condition $x(t)$ at time $t$ for the modified dynamics. Then, the PDE \eqref{pde2} has a feasible solution 
\begin{equation}
\label{T(x,t)}
T^0(x,t) = \int_0^t e^{A(t-s)} H(\breve X(s;x,t),s)ds,
\end{equation}
which is well-posed from the backward complete property. Therefore,
$
\phi^0(x,t) = x+ T^0(x,t)
$
is a feasible solution to \eqref{pde2} in $\mathtt{cl}(\mathcal{X})$. 

The remainder of the proof is to verify \eqref{rank_phi} under the additional assumption \eqref{jacobian_f}. The Jacobian of $\phi^0(x,t)$ is
$$
{\partial \phi^0 \over \partial x}(x,t)= I_n + {\partial T^0 \over \partial x}(x,t).
$$
Noting that
\begequ
\label{partial_H}
H(x_\star,t) = -f(x_\star,t) + A x_\star =0, ~
{\partial H \over \partial x}(x_\star,0) =0
\endequ
it yields that $\phi^0(x,t)$ is an immersion in $B_\varepsilon(x_\star)$ for small $\varepsilon>0$, if $k>0$ is sufficiently small. Following the similar proof of Theorem \ref{thm:inv1}, we redesign the mapping $\phi^0(x,t)$ as
$$
\phi(x,t) = e^{-A(t_x -t)} \phi^0(X(t_x;x,t), t_x)
$$
by choosing a sufficiently large $t_x>0$. Hence, we have verified the immersion condition.
\end{proof}

\begin{remark}\rm 
The additional assumption \eqref{jacobian_f} is relatively mild. A particular case is the system dynamics being in the form of $\dot x = f_1(x)+f_2(t)$, which satisfies \eqref{jacobian_f} automatically. In Proposition \ref{prop:inv2} and Theorem \ref{thm:conv_ltv}, we lift the contracting NLTV system \eqref{NLTV} into an LTI and an LTV system, respectively. In the latter, we remove the Jacobian requirement \eqref{jacobian_f} of the vector field $f(x,t)$. The key underlying reason relies on that in the latter we have ${\partial H \over \partial x}(x_\star,t) =0,~\forall t$, but for the former we only have ${\partial H \over \partial x}(x_\star,0) =0$ only at the initial moment without uniformity with respect to time.
\end{remark}

%


\section{Extension to Limit Cycles}
\label{sec5}

In this section, we show the equivalence between the transverse contraction and Koopman approaches in orbital stability analysis. Indeed, periodic behavior plays important roles in many engineering and biological applications \cite{YIORT,MANetal,SHIetal}. Given an autonomous system \eqref{NL:auto}, a (non-trivial) periodic solution $X$ is one for which there exists $T>0$ such that 
$$
X (t) = X(t+T), \quad t\ge 0,
$$
and the orbit is the set 
$$
\gamma:= \{x \in \rea^n | x =X(t), \; 0 \le t \le T\}.
$$

Analogous to the results for equilibria, in \cite[Proposition 3]{MAUMEZ} the authors propose a Koopman operator-based stability criterion for limit cycles. Let us recall the results in \cite{MAUMEZ}, and we modify as follows. Without loss of generality, we assume the system invariant in the set $\cal X$.

\begin{lemma}\rm
\label{lem:trans_koopman}
 Consider the system \eqref{NL:auto} with a hyperbolic limit cycle $\gamma \subset \rea^n$. If there exists a Koopman mapping 
$
\phi(x):=[\phi_{\lambda_1}(x), \ldots, \phi_{\lambda_{n-1}}(x)]^\top
\in \calo^{n-1}
$
satisfying
\begin{itemize}
\item[\bf T1] ({\em immersion}) $\Phi(x)= {\partial \phi \over \partial x}(x)$ is full row rank uniformly in the set $\gamma \subset \mathcal{X} \in \rea^n$, and $\phi(x)|_{x\in \gamma} = 0$;
\item[\bf T2] ({\em stability}) The existence of a Hurwitz matrix $A \in \rea^{(n-1)\times (n-1)}$ verifying the PDE \eqref{pde:new}.
\end{itemize}  
Then, the system is orbitally asymptotically stable in $\mathcal{X}$ with respect to the limit cycle $\gamma$. \QED
\end{lemma}

Due to the topological constraint, the set $\mathcal{X}$ cannot be the entire space of $\rea^n$. The best result we may get in the Euclidean space is that the basin of attraction is almost global, except a zero Lebesgue measure set.

Meanwhile, \cite[Thm. 3]{MANSLOscl} provides a transverse contraction criterion for existence and stability of a limit cycle. See Appendix for the definition of transverse contraction. 

\begin{lemma}
\rm \label{lem:trans_contraction} \cite{MANSLOscl}
Consider the system \eqref{NL:auto}. If there exists a uniformly bounded metric $M(x) \in \rea^{n\times n}_{\succ 0}$ such that
\begin{equation}
\label{trans_contraction_scl}
\dot M + {\partial f(x) \over \partial x}^\top M + M{\partial f(x) \over \partial x} - \rho(x) f(x)f(x)^\top \prec 0
\end{equation}
for some scalar function lower bounded $\rho(x)>0$, then the system is orbitally asymptotically stable.
\end{lemma}

The interested reader may refer to \cite{MANSLOscl} for its convex representation. A benefit of Lemma \ref{lem:trans_contraction}, compared to Lemma \ref{lem:trans_koopman}, is that it does not require the prior knowledge of the orbit $\gamma$ to verify orbital stability of a given nonlinear system. In \cite{YIetal}, it was also suggested to use the bounded \emph{semi-definite} Riemannian metric $\mathbb{M}(x) \in \rea^{n\times n }_{\succeq 0}$ to verify the transverse contraction with the inequality
\begin{equation}
\label{ineq:trans}
\dot \bbm(x) + \bbm(x) {\partial f(x) \over \partial x}  + {\partial f(x) \over \partial x}^\top \bbm(x)  \prec - k \bbm(x),\quad k>0.
\end{equation}
Since $\bbm$ can be parameterized as $\bbm(x) = \Psi(x) P(x) \Psi^\top(x)$ with $P(x) \in \rea^{r\times r}_{\succ 0}$ and $\Psi \in \rea^{n\times r}$, if $\nabla \Psi_i = (\nabla \Psi_i)^\top$ for all $i=1,\ldots, r$, then the attractive orbit can be obtained as
$$
\gamma = \{x \in \rea^n ~|~ \psi(x) = 0\},
$$
with 
$$
\psi(x):= \int_0^1 (\Psi(sx))^\top x ds
$$
if $\psi(x)|_\gamma =0$ with $r=n-1$.

\subsection{Koopman Implies Transverse Contraction}
\label{sec:k2tc}

We are now ready to show that the Koopman approach for limit cycles implies transverse contraction, {\em i.e.}, verifying both \eqref{ineq:trans} and the conditions in Lemma \ref{lem:trans_contraction}.

\begin{theorem}
\label{koopman2trans_contraction}\rm
Consider the system \eqref{NL:auto} satisfying the assumptions in Lemma \ref{lem:trans_koopman}. Then, the system is transversely asymptotically contracting with respect to $\phi$, {\em i.e.}, there exist 
\begin{itemize}
    \item[(1)] a metric $\bbm(x) \succeq 0$ satisfying \eqref{ineq:trans} globally;
    \item[(2)] a metric $M(x) \succ 0$ satisfying \eqref{trans_contraction_scl} locally.
\end{itemize}
\end{theorem}
\begin{proof}
If there exists a Koopman mapping $\phi$ satisfying the conditions {\bf T1}-{\bf T2}, we can find a matrix $P \in \rea^{r\times r}_{\succ 0}$ with $r:= n-1$ satisfying the Lyapunov inequality
$$
P A + A^\top P \preceq - k P, \quad  k>0.
$$
Then, we construct the semi-definite Riemannian metric as
$
\bbm(x) = {\partial \phi(x) \over \partial x}^\top P {\partial \phi(x) \over \partial x}.
$
Following the similar procedure as in the proof of Theorem \ref{thm:1}, we have
$$
\dot \bbm + \bbm F + F^\top \bbm \preceq  - k \bbm ,
$$
thus verifying the first sufficient condition in \eqref{ineq:trans}.

We briefly summarize the construction of the transverse contraction metric to guarantee the conditions in Lemma \ref{lem:trans_contraction} locally. Since ${\Phi(x)}= {\partial \phi \over \partial x}(x)$ is full rank, invoking Wazewski theorem \cite[Ch. 9.3]{BER}, we know that the Jacobian completion of $\nabla \phi$ is solvable, {\em i.e.}, there exists a $C^\infty$ mapping $\Theta: \rea^n \to \rea^{1\times n}$ such that 
\begequ
\label{full_rank:N}
\det (N(x))  \neq 0, \quad 
N(x):= \col(\Phi(x), \Theta(x))
\endequ
in any contractible sets.\footnote{If $\Theta(x)$ is additionally integrable, {\em i.e.} ${\partial \theta \over \partial x} = \Theta$, then $\phi(x)$ and $\theta(x)$ represent the transverse and tangential coordinates ({\em i.e.} isochrons, or angular variable), respectively, for orbital stability analysis.} Then, we may select the positive definite metric $M(x)$ as
$$\small
M(x) = N(x)^\top \begmat{P & 0 \\ 0 & 1} N(x) = \Phi(x)^\top P \Phi(x) + \Theta(x)^\top \Theta(x),
$$
in which $P$ is the solution to the Lyapunov equation $PA + A^\top P + I =0$. According to the results in \cite[Thm. 4]{MANSLOscl}, it completes the proof. 
\end{proof}

\subsection{Transverse Contraction Implies the Koopman Condition}

Now let us show the converse result, {\em i.e.}, transverse contraction implies the Koopman conditions for limit cycles. We make the following assumption.

\begin{assumption}
\label{assumption:limit}\rm
Consider a limit cycle $\gamma \in \mathcal{X} \subset  \mathbb{R}^{n} $ with the transverse coordinate $\xi \in \rea^{n-1}$ and the tangential coordinate $\theta \in \mathbb{S}$ defined by 
$$
\begmat{\xi \\ \theta}
= \phi_1(x):= \begmat{\phi_\xi(x) \\ \phi_\theta(x)},
$$
satisfying
$$
\phi_\xi(x)\Big|_{x\in \gamma} =0.
$$
Additionally, $\nabla \phi_1(x)$ is full rank for $x\in \mathcal{X}$.
\end{assumption}

\begin{theorem}
\rm\label{prop:converse_limit}
Consider the system \eqref{NL:auto}, which admits a (non-trivial) limit cycle $\gamma$ satisfying Assumption \ref{assumption:limit}. If there exists a uniformly bounded matrix $P(x) \in \rea^{(n-1)\times (n-1)}_{\succ 0}$ such that the semi-definite Riemannian metric $\bbm(x) = \nabla \phi_\xi(x) P(x) (\nabla \phi_\xi(x))^\top $ satisfying \eqref{ineq:trans}, then there is a Koopman mapping $\phi:\rea^n \to \rea^{n-1}$ satisfying {\bf T1} and {\bf T2}.
\end{theorem}


\begin{proof}
From the existence of a limit cycle, there are a family of particular solutions $x_\star(t)$ invariant on $\gamma$, {\em i.e.}
$$
x_\star(t) = x_\star(t+T) \in \gamma, \quad \forall t\ge 0
$$
for some $T>0$. According to \cite[Proposition 11]{YIetal}, we obtain
$$
\lim_{t\to\infty}\|\phi_\xi(X(x,t)) \|_\gamma =0, \quad \forall x\in \mathcal{X}
$$
with $\|x\|_\gamma := \inf_{y \in \gamma} |x- y|$ and the $\xi$-system dynamics
\begin{equation}
\label{dot_xi}
\dot\xi = {\partial \phi_\xi \over \partial x } f(x) \Big|_{x= \phi_1^{-1}(\xi,\theta)}
\end{equation}
is contracting in terms of the assumption \eqref{ineq:trans} and has an equilibrium at the origin, where $\phi_1^{-1}$ is the inverse mapping of $\phi_1$. According to Theorem \ref{thm:conv_ltv}, there exists a transformation $\phi_2: \rea^{n-1} \times \rea \to \rea^{n-1}$ lifting \eqref{dot_xi} into\footnote{The angular variable $\theta$ plays the role of time $t$ in Theorem \ref{thm:conv_ltv}. Indeed, we may find a reversible function between $(\xi,\theta)$ and $(\xi,t)$; see for example the proof of \cite[Thm. 2.6]{LANMEZ}.}
\begin{equation}
\label{dot_zeta}
\dot \zeta = A(\theta) \zeta,
\end{equation}
where $A$ is a stable (time-varying) matrix.

Note that $\theta$ is defined on $\mathbb{S}$ (via normalization), as well as $\dot \theta|_{\gamma} \neq 0$, and thus $A(\theta)$ is a periodic, stable matrix. According to Floquet theorem, there is a periodic matrix $Q(\theta)$ such that the transformed coordinate $z = Q(\theta) \zeta$ has an LTI dynamics
$$
\dot z = A_Q z
$$
with $A_Q \in \rea^{(n-1)\times (n-1)}$ Hurwitz.

As a result, the Koopman mapping is given by the composite function
$$
\phi(x) := Q(\phi_\theta(x))\cdot \phi_2 \circ \phi_\xi(x),
$$
completing the proof.
\end{proof}

In the above result, we assume that we already know the transformation to get transverse and tangential coordinates. Indeed, the existence of such a transformation in its domain of attraction was shown in \cite{BYR}. The converse result in Theorem \ref{prop:converse_limit} resembles the necessary part of \cite[Proposition 3]{MAUMEZ}.

\section{Discussions}
\label{sec6}

\subsection{Extension to Stabilization Problems}
\label{sec51}

We have showed the equivalence between the Koopman and contraction approaches when analyzing stability of equilibria, trajectories, and limit cycles. On the other hand, it is of practical interest to study if such equivalence holds in constructive problems.

Let us consider the controlled system model
\begin{equation}
\label{NLCS}
\dot x = f(x,u),
\end{equation}
with the control input $u \in \rea^m$.

\begin{definition}\label{def:ccm}\rm 
({\em control contraction metric}) Consider the controlled system \eqref{NLCS} with Jacobians $F(x,u):={\partial f(x,u)\over \partial x}$ and $G(x,u):={\partial f(x,u) \over \partial u}$. If we can find a uniformly bounded metric $M(x)$ and a function $K(x)$ satisfying
\begin{equation}
\label{ccm}
\begin{aligned}
	\dot M + MF+F^\top M + MGK+(GK)^\top M &\prec 0,
\end{aligned}
\end{equation}
then we call $M(x)$ a (strong) control contraction metric.
\end{definition}

For affine-in-input systems, the CCM may be written in a more compact way independent of $K(x)$. The interested reader may refer to \cite{MANSLO} for additional details.


Similar to the autonomous case, for the controlled system \eqref{NLCS} we may define the Koopman mapping $\phi$ satisfying 
\begin{equation}
\label{pde:ctrl}
    {\partial \phi \over \partial x}(x) f(x,u) = A\phi(x) + B u, \; \forall u \in \rea^m,
\end{equation}
with matrix $B \in \rea^{n\times m}$, which transforms the dynamics into
\begequ
\label{z:lti}
\dot z = Az + Bu, \; z(0)= \phi(x(0)).
\endequ
We have the following.

\begin{proposition}
\rm \label{prop:ccm}
Consider the controlled system \eqref{NLCS}, which has a Koopman mapping $\phi$ satisfying {\bf C1} and the PDE \eqref{pde:ctrl}. If the lifted LTI system \eqref{z:lti} is stabilizable, then the given controlled system admits a CCM.
\end{proposition}
\begin{proof}
The stabilizability of the LTI system \eqref{z:lti} is equivalent to the existence of a matrix $P = P^\top \succ 0$ and a feedback gain matrix $\bar K$ such that
\begin{equation}
\label{lyp_ctrl}
P(A+B\bar K) + (A+ B\bar K)^\top P \prec 0.
\end{equation}

Now we define the feedback controller 
$
u = \bar K\phi(x),
$
and substitute into the PDE \eqref{pde:ctrl}, obtaining
$$
\Phi(x) f(x, \bar K\phi(x)) = (A + B\bar K) \phi(x).
$$
Calculate its partial derivative with respect to $x$, yielding
$$
\dot \Phi(x) + \Phi(x) [F+ G\bar K \Phi(x)] = (A+ B\bar K) \Phi.
$$
By selecting the metric
$$
M(x):= \Phi(x)^\top P \Phi(x)
$$
and the mapping $K(x):= \bar K \Phi(x)$, we have
$$
\begin{aligned}
 & \quad \Phi^\top [P(A+B\bar K) + (A+B \bar K)^\top P]\Phi
 \\
 = & \quad 
 \dot M + MF + F^\top M + MGK + (GK)^\top M 
 \\
\prec& \quad  0,
\end{aligned}
$$
which exactly coincides with the strong CCM in \eqref{ccm}.
\end{proof}

\begin{remark}
\rm
Although Koopman and CCM methods for control design are equivalent in certain cases, they differ in their implementation and each have their advantages. In particular, CCM methods admit a convex search for the metric and the differential controller $K$, however there is no guarantee that the obtained $K(x)$ is integrable and some online computation may be required to realize a specific controller \cite{LEUMAN,WANMAN}. On the other hand, the joint search for observables $\phi$ and controller in the Koopman framework is non-convex, but if successful there always admits an explicit controller $u=K\phi(x)$.
\end{remark}

\begin{remark}
\rm
In general, the lifted $z$-dynamics may contain a state-dependent input matrix rather than constant $B$ \cite{HUAetal,GOSPAL}. Here, we use the constant input matrix assumption to simplify the presentation, as is popular in data-driven Koopman-based methods for control, {\em e.g.} \cite{PROetal}, \cite[Ch. 5]{MAUetal}. A straightforward extension is the case with state-dependent $B(x)$, and in many cases we may parameterize $u= \alpha(x,v)$ and regard $v$ as the new input in order to get the form $\dot z = Az + Bv$. On the other hand, noting that the input $u$ can be arbitrary, the PDE \eqref{pde:ctrl} imposes implicitly the affine-in-control assumption of $f(x,u)$, making it relatively restrictive. It is also promising to be extended to more general cases.
\end{remark}


\begin{remark}
It is interesting to study the converse claim for stabilization problems. For an affine-in-control system $\dot x = f(x) + gu$, if there exists a feedback law $u = k(x) +v(t)$ making the closed-loop system contracting, then following the main results of the paper, there is a mapping $\phi(x)$ which lifts the closed-loop into a stable linear system, and also lifts the open loop into a stabilizable dynamics. Unfortunately, it is not a \emph{bona fide} converse, since a CCM controller may be not integrable to a function $k(x)$ with $K(x) = {\partial k \over\partial x}(x)$.
\end{remark}

\subsection{Relations to KKL Observers}
\label{sec62}

The KKL observer, also known as nonlinear Luenberger observer, is a state estimation approach for general nonlinear systems \cite{KAZKRA}, which consists in mapping the given nonlinear dynamics 
$$
\dot x = f(x,t) , \quad y= h(x,t)
$$
to a linear system
$$
\dot z = Az + Dy
$$
with $A$ Hurwitz, in which the output $y$ is viewed as an available signal. This step relies on solving the PDE
\begin{equation}
\label{pde:kklo}
{\partial T \over \partial t} + {\partial T \over \partial x}(x) f(x,t) = AT(x,t) + Dh(x,t).
\end{equation}
It is a quite general framework since it allows a nonlinear output injection $h(x,t)$ to appear in the transformed $z$-coordinate. The associated PDE \eqref{pde:kklo} is always solvable under some mild technical assumptions, and when applying KKL observer, the main task is to guarantee the coordinate change $x\mapsto z= T(x,t)$ injective, thus admitting an inverse mapping. It has been shown that, the general notion -- \emph{backward distinguishability} -- is sufficient to guarantee the injectivity \cite{ANDPRA,BER}. Some remarks about the relation between the Koopman method and KKL observers are in order.
\begin{itemize}
\item[1)] In the proof of the converse results in Sections \ref{sec4}, we parameterize the coordinate change $\phi: (x,t)\mapsto z$ into two parts, {\em i.e.}, $\phi(x,t) = x +T(x,t)$, and then we obtain \eqref{pde:kkl} for autonomous systems and \eqref{pde3} for time-varying systems. It exactly coincides with the PDE involved in KKL observers for the auxiliary system
\begin{equation}
\label{aux:aut}
\dot x = f(x,t), \quad y = H(x,t).
\end{equation}
with $H(\cdot)$ the high-order remainder terms defined in Section \ref{sec4}. Using the constructive solution in KKL observers, we get a feasible solution to the PDE in the Koopman methods for nonlinear contracting systems.
\item[2)] The main difference, between the Koopman method and KKL observers, relies on how to guarantee the injectivity. For the latter, we need to show the injectivity of $T(x,t)$ by exploiting the backward distinguishability of the given system. However, for the Koopman method, we need to show the injectivity of $\phi(x,t) = x + T(x,t)$ rather than $T(x,t)$ itself. Our key idea is to utilize its ``identity part''.

\item[3)] In KKL observers, it is suggested to use excessive coordinates, generally more than $(2n+1)$-dimensional\footnote{For elements selected in $\mathbb{C}$, the dimension is not less than $n+1$.}, in order to get the injectivity of the mapping $T$ defined in \eqref{T(x,t)}. On the other hand, excessive coordinates are widely adopted for the Koopman operator in the learning literature. It is claimed in \cite[Thm. 3]{KORMEZ} that by choosing sufficiently rich orthonormal bases, the solution of a least square approximates the Koopman operator with guaranteed accuracy.
\end{itemize}


\begin{remark}\rm 
(\emph{Extensions to control design}) An interesting open problem is the converse claim of Proposition \ref{prop:ccm}, which is related to the dual problem of KKL observer. A straightforward idea is to immerse the controlled system \eqref{NLCS} into
$$
\dot z = Az + h(x,u), \qquad z=T(x)
$$
for some function $h(x,u)$ to be determined, with $A$ stable. In this step, we have the same PDE as the one in KKL observers. Then, the stabilization task generally contains two tasks:
\begin{itemize}
    \item[-] finding a function $\alpha(x)$ to solve the algebraic equation $h(x,\alpha(x))=0$;
    \item[-] the function $h(x,t)$ should guarantee that the system $\dot x = f(x,u),\; y=h(x,u)$ is backward distinguishable.
\end{itemize}
Then, the feedback law $u=\alpha(x)$ stabilizes the system at some equilibrium. 
\end{remark}

\section{Examples}
\label{sec7}
\subsection{A 2-Dimensional System}
\label{sec:example1}

Consider the nonlinear autonomous system \cite{BRUetal}
\begequ
\label{example1}
\dot x  = \begmat{
 - x_1 \\ - x_2 + x_1^2
 }
\endequ
with $x\in \rea^2$. The differential dynamics of \eqref{example1} is given by
$
\delta \dot x =F(x) \delta x,
$
with the Jacobian
$$
F(x)= \begmat{-1 & 0 \\ 2x_1 & -1}.
$$
By selecting the metric $M(x) = \diag(1+4x_1^2, 1)$, we may verify that the given system is contracting due to
$$
\dot M(x) + M(x)F(x) + F(x)^\top M(x) 
=
\begmat{-2 -16 x_1^2 & 2x_1 \\ 2x_1 & -2} \prec 0.
$$

This example will be used to verify the converse result in Section \ref{sec32}. The system has an equilibrium at the origin, {\em i.e.} $x_\star = 0$, and then following the proof of Theorem \ref{thm:inv1} we have
$
F_\star = F(x_\star) = \diag(-1,-1).
$
A feasible transformation is $\phi(x) = x+T(x)$, with $T(x)$ the solution of \eqref{T}. The flow $X(x,t)$ of the given nonlinear system \eqref{example1} can be obtained as
$$
X(x,t) 
=
\begmat{e^{-t} x_1
\\ 
e^{-t} x_1^2 + e^{-t} x_2 - e^{-2t}x_1^2
},
$$
and the high-order remainder term is
$
H(x) = \col(0 , -x_1^2).
$
We consider the modified system \eqref{modified} in the the open set $\mathcal{X}:=\{0< x_1<1\}$, and the backward flow of $\breve X_1$ is given by
$$
\breve X_1(x,t) = \left\{
\begin{aligned}
& e^{-t} x_1, ~& \ln x_1 \le t \le 0 
 \\
& 1,~ &  t < \ln x_1.
\end{aligned}
\right.
$$
Then, it yields
$$
\begin{aligned}
T(x) & = \int_0^{+\infty} \exp(  F_\star s) H(\breve X(x,-s))ds
\\
& = \bigintssss_{-\infty}^0 \begmat{e^s & \\ & e^s} \begmat{0 \\ - \breve X_1(x,s)^2} ds
\\
& =
- \bigintssss_{\ln x_1}^0
 \begmat{0 \\ e^{s} (e^{-s} x_1)^2} ds
 -
 \bigintssss_{-\infty}^{\ln x_1}
 \begmat{0 \\ e^{s} 1^2} ds
\\
& 
=
\begmat{0 \\ - 2x_1 + x_1^2}
\end{aligned}
$$
thus
$$
\phi(x) = \begmat{x_1 \\ -2x_1 + x_1^2 + x_2}.
$$
It is straightforward to verify {\bf C1} and {\bf C2}, {\em i.e.}, ${\partial \phi \over \partial x}(x)$ is full rank, and
$$
\dot{\aoverbrace[L1R]{\phi(x)}}
=
\begmat{-x_1\\  2x_1 - x_1^2 - x_2}
=
F_\star \phi(x).
$$
We underline here that these conditions hold globally in $\rea^2$. It is interesting to compare the above result with the one in \cite{BRUetal}, in which the given system is immersed into a three-dimensional LTI system by introducing excessive coordinates.

\subsection{Limit Cycle in Induction Motor}

In this section, we use the model of an induction motor to illustrate the result in Section \ref{sec:k2tc}, {\em i.e.}, the Koopman condition implies transverse contraction for limit cycles. The normalized model in the fixed frame is given by \cite{ORTbook}
\begin{equation}
\label{im}
\begin{aligned}
\dot{\psi}_r & = - R\phi_r + \omega \mathbb{J} \psi_r + Ru
\\
\dot{\omega} & = u^\top \mathbb{J} \psi_r  - \tau_L, 
\quad 
\mathbb{J}:=\begmat{0 & - 1\\ 1 & 0},
\end{aligned}
\end{equation}
with the flux $\psi_r \in \rea^2$, angular speed $\omega \in \rea$, the load torque $\tau_L \in \rea$, the resistance $R>0$, and the stator current $u\in \rea^2$ as input. A basic control problem is to regulate the norm $|\psi_r|$ and the speed $\omega$ to some constants $\beta_\star$ and $\omega_\star$, respectively. To address this, the classical field-oriented control (FOC), which was introduced in the drives community in 1972 \cite{BLA}, is now the \emph{de facto} standard in all high-performance applications of electric drives. With zero load $\tau_L$, the FOC takes the form
\begin{equation}
\label{FOC}
u = \Big[ \beta_\star I_2 - {k \over \beta_\star} (\omega-\omega_\star) \mathbb{J} \Big] {\psi_r \over |\psi_r|}, \quad k>0.
\end{equation}

We assume that -- with loss of generality -- all constant parameters and gains being one to simplify the presentation. In \cite[Proposition 4]{YIORT} it shows that the FOC \eqref{FOC} achieves \emph{almost global} orbital stabilization. The following gives an alternative proof from the perspective of Koopman operator.

\begin{proposition}
The induction motor model \eqref{im} in closed loop with \eqref{FOC} satisfies the assumptions in Lemma \ref{lem:trans_koopman}, and thus has an attractive limit cycle.
\end{proposition}
\begin{proof}
For consistency of notations, we define the state $x = [\psi_r^\top, \omega]^\top \in \rea^3$ with $x_p = [x_1,x_2]^\top$. Then, the closed loop is given by
\begin{equation}
\label{IM:closed_loop}
\dot x 
= 
\begmat{
1 - |x_p| & - x_3|x_p| & {x_2\over |x_p|}
\\
* & 1 - |x_p| & - {x_1 \over |x_p|}
\\
* & * & -|x_p|
}
\nabla \calh
\end{equation}
with ``*'' presenting some skew-symmetric elements, and the Hamiltonian $\calh(x) = \hal |x_p|^2 +\hal (x_3-1)^2$. Now, let us verify the assumptions for the closed loop \eqref{IM:closed_loop}. It is straightforward to verify that the Jordan curve
$
\gamma = \{x\in \rea^3 : |x_p| =1, x_3 = 1\}
$
is forward invariant, and there is no equilibrium on the set. Let us select the Koopman eigenfunctions as
$$
\begin{aligned}
\phi_1(x) = 1-{ 1\over |x_p|}
, \quad
\phi_2(x)  = {
x_3  - 1 \over |x_p|} .
\end{aligned}
$$
The Jacobian ${\partial \phi \over \partial x}(x) $ is full rank almost globally, except the zero-Lebesgue measurable set $\Omega_s:=\{x\in \rea^3 : x_1=x_2=0\}$, verifying the condition {\bf T1}. On the other hand, we have 
$$
\dot{\aoverbrace[L1R]{\phi(x)}} ~=~ - \phi(x),
$$
thus verifying {\bf T2}, {\em i.e.}, the PDE \eqref{pde:new} with the Hurwitz matrix $A = - I_2$. Invoking Lemma \ref{lem:trans_koopman}, we complete the proof.
\end{proof}

According to Theorem \ref{koopman2trans_contraction}, there exists a transverse contraction metric $M(x)$ satisfying the inequality \eqref{trans_contraction_scl}. To see this, we may select the matrix-valued function 
$
\Theta(x) = \begmat{-x_2 & x_1 & 0}
$
to guarantee the full-rank condition \eqref{full_rank:N} almost globally except the singular set $\Omega_s$. By selecting $P= I_2$ we have $PA+ A^\top P \prec 0$, and as a result the matrix 
$$
M(x) = \mbox{$\big[{\partial \phi \over \partial x}(x)\big] $} P \mbox{$\big[{\partial \phi \over \partial x}(x)\big]^\top $}  + \Theta(x)^\top \Theta(x)
$$
is positive definite except $\Omega_s$. 
\begin{figure}[!htp]
    \centering
    \includegraphics[width=0.7\linewidth]{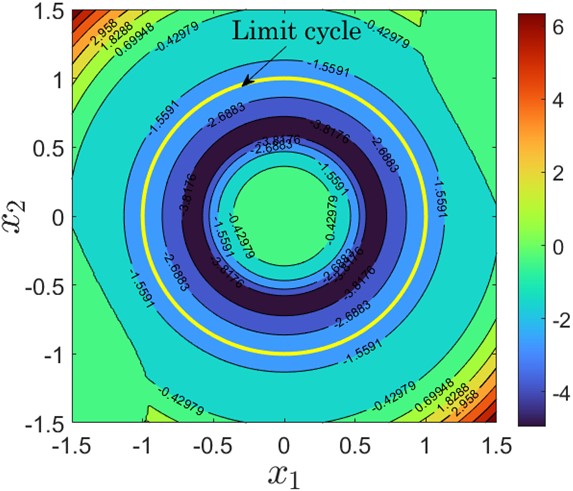}
    \caption{The largest real parts of the eigenvalues of $\dot M + {\partial f \over \partial x}^\top M  + M {\partial f \over \partial x} - \rho f f^\top$ via fixing $x_3 = 1$}
    \label{fig:limit}
\end{figure}
In Fig. \ref{fig:limit}, we draw the largest real parts of the eigenvalues of $\dot M + {\partial f \over \partial x}^\top M  + M {\partial f \over \partial x} - \rho f f^\top$ with $\rho = 50$, and note that for the purpose of visualization we fix $x_3$ at its steady-state value $x_3=1$. Clearly, $M(x)$ verifies the transverse contraction condition \eqref{trans_contraction_scl} in the neighborhood of the limit cycle.

\begin{remark}
In the above analysis, we use explicitly the equation of the limit cycle $\gamma$ to illustrate the theoretical results in Section \ref{sec:k2tc}. However, it is unnecessary to write down the analytic form of $\gamma$ via transverse contraction to find a limit cycle and prove its stability \cite{MANSLOscl}. This is particularly important for showing robustness of a limit cycle, since typically a limit cycle location will change as parameters change.
\end{remark}

\subsection{Learning Contraction Metrics from Data}

It is widely recognized that the Koopman operator provides a powerful tool to learn dynamical models of nonlinear systems from data \cite{KORMEZ,MAMetal,MACetal,MAUetal}. Based on the equivalence between contraction and Koopman approaches studied in the paper, it can provide a novel approach to learn contraction metrics for stable nonlinear systems from pure trajectory data. Learning contraction metrics has recently been explored in the context of robust motion planning and control, {\em e.g.} \cite{SUNetal,CHOetal,TSUetal}, and our results dramatically simplify this problem to one of linear system identification. 

{\em Problem.} ({\em Data-driven contraction metrics learning}) Consider a contracting system \eqref{NL:auto} and assume that only a set of state trajectory data $\{\tilde x(k)\}_{k=0}^T$ and the derivatives $\{\dot{\tilde x}(k)\}_{k=0}^T$ are available\footnote{Sometimes the time derivative of systems state is not available from sensors. We may apply stable filters to obtain new regressors.}, and the vector field $f(x)$ is unknown. Our task is to estimate the contraction metric $M(x)$ using the information of data only.

According to Theorem \ref{thm:inv1}, if the system is contracting we may always find a Koopman mapping $\phi$ to get a lifted LTI system \eqref{dot_z} with $A$ Hurwitz, and $\phi$ is left invertible. A contraction metric of the given system is
$
M = \nabla \phi P (\nabla \phi)^\top
$
with $P \succ 0$ the solution of the Lyapunov equation. Based on this intuitive idea, data-driven contraction metrics learning is translated to the problem of estimating the Koopman mapping $\phi$ and the stable matrix $A$ from data.

First, selecting the basis function 
$
w(x) \in \rea^N,
$
with $N \in \mathbb{N}$, and here $N>n$ is usually selected sufficiently large in order to get high accuracy. Then, we are able to parameterize the Koopman mapping as
\begequ
\phi(x) = \theta^\top  w(x), \quad \theta \in \rea^{N \times n}.
\endequ
Now the PDE \eqref{pde:new} becomes
\begequ
\label{theta_w}
\theta^\top \dot{\aoverbrace[L1R]{w(x)}}   = A\theta^\top w(x)
\endequ
with $\dot w (x) = W(x) \dot x$ and $W(x):={\partial w \over \partial x}(x)$.
Our target, then, becomes searching for $\theta$ and $A$ in the optimization problem
\begequ
\label{min1}
\begin{aligned}
\underset{\theta,A}{\min} &\quad \sum_{k=1}^{n_k}\left| \theta^\top Y_1(k) - A \theta^\top Y_2(k) \right|^2
\\
\mbox{s.t.}& \quad \rank{\theta} = n
\end{aligned}
\endequ
with the measurable vectors $Y_1(k) := W(\tilde x(k))\dot{\tilde{x}}(k)$  and $Y_2(k) := w(\tilde x(k))$.

We consider the following assumption. If we select sufficient numbers of independent basis functions $(N\gg n)$, all the elements of $W(x)f(x)$ can be approximately linearly represented by the bases $w(x)$, {\em i.e.}
$
[W(x)f(x)]_{i} = \rho^\top_{i} w(x) + O(\cdot), \; \forall i
$
for some vectors $\rho_{i}$, with some (tiny) high-order term $O(\cdot)$. Since $\theta$ is full rank, we can always find another matrix $\theta_\bot \in \rea^{N\times (N-n)}$ such that $[\theta, ~\theta_\bot]$ is full rank. Hence,
\begin{equation}
\label{theta_bot}
\begmat{\theta^\top \\ \theta_\bot^\top}  \dot{\aoverbrace[L1R]{w(x)}} 
=
\begmat{A\theta^\top \\
\theta_\bot^\top\begmat{\rho_1 & \cdots & \rho_{N-n}}^\top 
}
w(x) + O(\cdot),
\end{equation}
and then we have
$$\footnotesize
\dot{\aoverbrace[L1R]{w(x)}}  = \tilde{A} w(x) + O(\cdot),
~
\tilde{A}:= \begmat{\theta^\top \\ \theta_\bot^\top}^{-1}
\begmat{A\theta^\top \\
\theta_\bot^\top \begmat{\rho_1 & \cdots & \rho_{N-n}}^\top
}.
$$
Now we may estimate the matrix $\tilde{A}$ using the bases $w(x)$ as observables to identify the dynamical model. This is the underlying reason why, in many works, only (sufficiently high-dimensional) lifted system matrix $\tilde{A}$ is learned -- rather than estimating the Koopman mapping $\phi$ and matrix $A$ simultaneously -- but it still provides satisfactory identification results. Then we may use the least square solution to estimate $\tilde{A}$ as
$
\hat{A} = \tilde Y_1 \tilde Y_2^\top (\tilde Y_2 \tilde Y_2^\top)^{-1}
$
with $\tilde Y_1 := \begmat{Y_1(1) & \cdots & Y_2(k)}$ and $\tilde Y_2 := \begmat{Y_2(1) & \cdots & Y_2(k)}$. After solving $\hat Q \hat A + \hat A^\top \hat Q = -I_N$, we learn a contraction metric 
$
M = {\partial w \over \partial x}^\top \hat Q {\partial w\over \partial x}.
$

{\em Numerical example.} Let us consider the system in Section \ref{sec:example1}. But, now we only have the dataset $\{\tilde x(k),\dot{\tilde x}(k)\}_{k=0}^T$ rather than the model itself. The basis function is selected as polynomials $w(x)= \col(x_1, x_2, x_1x_2, x_1^2, x_2^2)$. 
The trajectory data are collected under the sampling time 0.2s from the initial condition $[-2 ~~1]^\top$. Following the above, the estimated contraction metric $ M(x) \in \rea^{2\times 2}$ is obtained as
$$\small
 M(x) = 
W(x)^\top
\begmat{   0.45 &    0.23&  0.03& -0.04&  0.04\\
 *&    0.44& - 0.01& -0.13& -0.04\\
*& * &   0.23&  0.1&  0.12\\
*&  *&  *&  0.29&  0.07\\
*&  *&   *& *&  0.23}
W(x).
$$
We give in Fig. \ref{fig:simulation} the largest real part of the eigenvalues of $\dot M + MF + F^\top M$, which is negative definite -- showing that the Riemmanian metric is geodesically decreasing -- and the smallest real part of the eigenvalue of $M$, which is positive definite. Hence, the obtained estimation of $M$ is qualified as a contraction metric. It should be kept in mind that, for a given contracting system, the contraction metric is not unique.

\begin{figure*}[htp!]
    \centering
\subfloat[The largest (real parts of) eigenvalue of $(\partial_f M + MF + F^\top M)$]
{
\includegraphics[width=0.27\textwidth]{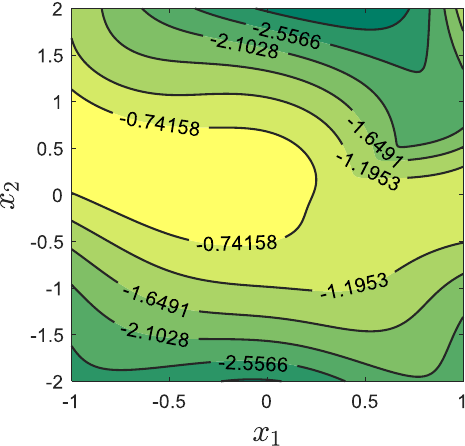}
\label{fig:3}
}
\quad
\subfloat[The smallest (real parts of) eigenvalue of $M$]
{
\includegraphics[width=0.27\textwidth]{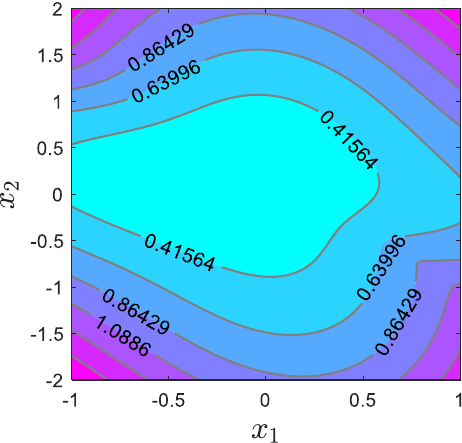}
\label{fig:4}
}
    \caption{Learned contraction metric from trajectory data}
    \label{fig:simulation}
\end{figure*}

\begin{remark}\rm 
We use this example to show an interesting byproduct of the main results in the paper, {\em i.e.}, the Koopman approach provides an efficient methodology to learn contraction metrics from trajectory data. Some remarks are in order.
\begin{itemize}
\item By using sufficiently high-dimensional basis observables, we try to learn the matrix $\tilde{A}$ in the $w$-coordinate, which admits a modelling error from the high-order approximation term. If we come back to \eqref{min1} and search for $A$ and $\phi$, it is promising to improve accuracy. It, however, is a non-convex optimization problem.

\item The above procedure is sensitive to the quantity and coverage of data, with the obtained contraction metric local to the trajectory data. A possible approach is to impose stability constraint of the matrix $A$ during identification, {\em e.g.} via the methods of \cite{Lacy03, maciejowskiGuaranteedStabilitySubspace1995, Umenberger:2018} towards robust learning algorithms. Recently, we have extended these equivalence results from continuous-time to discrete-time models in \cite{FANetal}, leading to a nonlinear system identification approach with the \emph{model stability constraint}. We refer the reader to see more practical examples therein.
\end{itemize}
\end{remark}

\section{Concluding Remarks}
\label{sec8}


In this paper, we study the connections among the Koopman method, contraction analysis and KKL observers. First, we show that the conditions in the Koopman method implies contraction of a given autonomous or time-varying system by substitution of differential equations. Later, the converse results for both equilibria and limit cycles are proven under mild technical assumptions, the success of which relies on the interesting observation that the Koopman method enjoys the same PDE as the one in KKL observers. We also studied the nonlinear controlled systems, showing that stabilizability of the lifted linear systems implies the existence of a CCM of the original nonlinear system. In terms of the above results, it is reasonable to expect a further integration of these methods as a systematic constructive tool in the near future.

\appendix
\subsection{Definition of Transverse Contraction}

The concept of transverse contraction was proposed in \cite{MANSLOscl} for analysis of limit cycles. In \cite{YIetal}, it was generalized to deal with observer design.

\begin{definition}
\label{def:trans_contraction}\rm 
The forward complete system \eqref{NL:auto} is said to be transversely contracting with rate $\lambda>0$ (or transversely asymptotically contracting) with respect to $\psi: \rea^n \to \rea^ n$ ($1\le r \le n$), if for any pair $(x_a,x_b)\in \rea^n \times \rea^n$, we have
$$
|\psi(X(x_a,t)) - \psi(X(x_b,t))| \le e^{-\lambda t} b(x_a,x_b)
$$
form some function $b\ge 0$ with $b(x,x)=0$ or
$$
|\psi(X(x_a,t)) - \psi(X(x_b,t))| \le \kappa(|x_a - x_b|,t)
$$
for transverse asymptotic contraction, with $\kappa $ of class $\mathcal{KL}$.
\end{definition}

The key difference between transverse contraction and horizontal contraction \cite{FORSEP} is that the they are defined on state-space and tangent space, respectively. When adopting transverse contraction to the stability analysis of limit cycles \cite{MANSLOscl}, it becomes the existence of a Finsler-Lyapunov function $V(x,\delta x)$ satisfying
$$
{\partial V(x,\delta x) \over \partial x} f(x) + {\partial V(x,\delta x ) \over \partial \delta x} {\partial f(x) \over \partial x}\delta x \preceq -\lambda V(x,\delta x),
$$
for all non-zero $\delta x$ such that ${\partial V \over \partial \delta x}f(x) =0$.

\section*{Acknowledgments}
The authors would like to express their gratitude to three reviewers for their thoughtful comments that helped improve its clarity, as well as Reviewer 5 of the conference version \cite{YIetal2} for clarification of Remark \ref{remark:C} and Proposition \ref{prop:koopman}.

\bibliographystyle{abbrv}
\bibliography{ref_koopman}

\begin{IEEEbiography}[{\includegraphics[width=1.1in,height=1.25in,clip,keepaspectratio,,clip,trim={.4in .8in .4in 0.4in}]{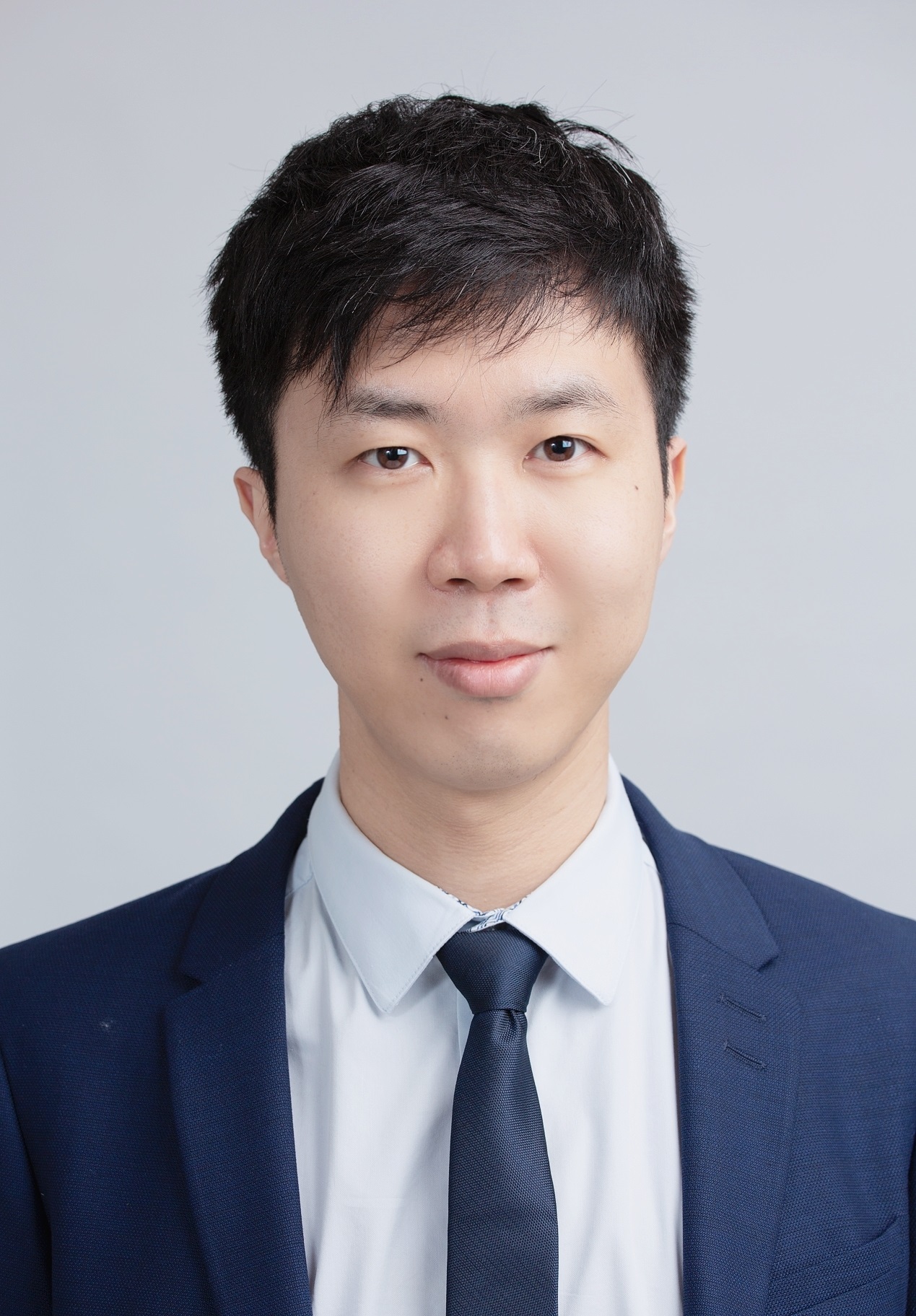}}]{Bowen Yi} 
obtained his Ph.D. degree in Control Engineering from Shanghai Jiao Tong University, China in 2019. From 2017 to 2019 he was a visiting student at Laboratoire des Signaux et Syst\`emes, CNRS-CentraleSup\'elec, Gif-sur-Yvette, France. He has held postdoctoral positions in Australian Centre for Robotics, The University of Sydney, Australia (2019 - 2022), and the Robotics Institute, University of Technology Sydney, Australia (Sept. 2022 - Oct. 2023). He is currently an Assistant Professor with the Department of Electrical Engineering, Polytechnique Montr\'eal, Canada. His research interests involve nonlinear systems and robotics. He received the 2019 CCTA Best Student Paper Award from the IEEE Control Systems Society for his contribution on sensorless observer design.
\end{IEEEbiography}
%

\begin{IEEEbiography}[{\includegraphics[width=1.1in,height=1.25in,clip,keepaspectratio]{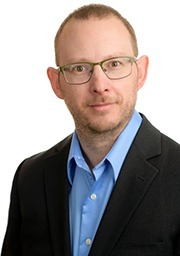}}]{Ian R. Manchester}
(M$'$18) received the B.E. (Hons 1) and Ph.D. degrees in Electrical Engineering from the University of New South Wales, Australia, in 2002 and 2006, respectively. He has held research positions at Ume\aa\ University, Sweden, and Massachusetts Institute of Technology, USA. In 2012 he joined the faculty at the University of Sydney, where he is currently Professor of Mechatronic Engineering, Director of the Australian Centre for Robotics (ACFR) and Director of the Australian Robotic Inspection and Asset Management Hub (ARIAM). His research interests include optimization and learning methods for nonlinear system analysis, identification, and control, and applications in robotics and biomedical engineering. He is an Associate Editor for IEEE Control Systems Letters and has previously served as an Associate Editor for IEEE Robotics \& Automation Letters.
\end{IEEEbiography}


\end{document}